\documentclass[hidelinks,onefignum,onetabnum]{siamart250211}

\usepackage{lipsum}
\usepackage{amsfonts}
\usepackage{graphicx}
\usepackage{epstopdf}
\usepackage{algorithmic}

\ifpdf
  \DeclareGraphicsExtensions{.eps,.pdf,.png,.jpg}
\else
  \DeclareGraphicsExtensions{.eps}
\fi

\newsiamremark{remark}{Remark}
\newsiamremark{hypothesis}{Hypothesis}
\crefname{hypothesis}{Hypothesis}{Hypotheses}
\newsiamthm{claim}{Claim}
\newsiamremark{fact}{Fact}
\crefname{fact}{Fact}{Facts}

\newsiamremark{assumption}{Assumption}
\usepackage{amssymb}

\headers{LQ Discrete-Time Dynamic Games with Unknown Dynamics}{S. Huang, X. Yang, Z. Cao and W. Mei}

\title{Linear-Quadratic Discrete-Time Dynamic Games with Unknown Dynamics}

\author{Shengyuan Huang\thanks{Academy of Mathematics and Systems Science, Chinese Academy of Sciences, Beijing 100190  China; University of Chinese Academy of Sciences, Beijing 100190, China (\email{huangshengyuan22@mails.ucas.ac.cn}).}
\and Xiaoguang Yang\thanks{Academy of Mathematics and Systems Science, Chinese Academy of Sciences, Beijing 100190  China; University of Chinese Academy of Sciences, Beijing 100190, China (\email{xgyang@iss.ac.cn}).}
\and Zhigang Cao\thanks{School of Economics and Management, Beijing Jiaotong University, Beijing, 100044, China 
  (\email{zgcao@bjtu.edu.cn}).} 
\and Wenjun Mei\thanks{Corresponding author. Department of Mechanics and Engineering
Science, Peking University, Beijing 100871, China
  (\email{mei@pku.edu.cn}).}}

\usepackage{amsopn}

\ifpdf
\hypersetup{
  pdftitle={Linear-Quadratic Discrete-Time Dynamic Games with Unknown Dynamics},
  pdfauthor={S. Huang, X. Yang, Z. Cao and W. Mei}
}
\fi

\begin{document}

\maketitle

\begin{abstract}
Considering linear-quadratic discrete-time games with unknown input/output/state (i/o/s) dynamics and state, we provide necessary and sufficient conditions for the existence and uniqueness of feedback Nash equilibria (FNE) in the finite-horizon game, based entirely on offline input/output data. We prove that the finite-horizon unknown-dynamics game and its corresponding known-dynamics game have the same FNEs, and provide detailed relationships between their respective FNE matrices. To simplify the computation of FNEs, we provide an invertibility condition and a corresponding algorithm that computes one FNE by solving a finite number of linear equation systems using  offline data. For the infinite-horizon unknown-dynamics game, limited offline data restricts players to computing optimal strategies only over a finite horizon. We prove that the finite-horizon strategy ``watching $T$ steps into the future and moving one step now,'' which is commonly used in classical optimal control, exhibits convergence in both the FNE matrices and the total costs in the infinite-horizon unknown-dynamics game, and further provide an analysis of the convergence rate of the total cost.  The corresponding algorithm for the infinite-horizon game is proposed and its efficacy is demonstrated through a non-scalar numerical example. 
\end{abstract}



\section{Introduction}
Dynamic games hold a significant position in game theory, as numerous application scenarios naturally involve competitions with multiple stages and dynamic interactions among players, such as economics \cite{res2020, aer2_2024, Econometrica_2025, aer1_2024}, management science \cite{ms1_2025,ms2_2025,or_2025}, control theory \cite{1998basarNoncooperativeGame,siam2024,automatica_game2025,TACgame2024}, robots \cite{natureMI_2024,scirobotics_2024,nature2019}, biological and social systems \cite{nc2024,nc2024_2,nc_2025}, and artificial intelligence \cite{nips2023,icml2021}.  Basar et al. \cite{1998basarNoncooperativeGame} provide a comprehensive review of information structures and equilibria in dynamic games. One of the most attractive equilibria is the feedback Nash equilibrium (FNE), where players act based on the continuously changing system state. However, the computation of FNEs in  infinite-horizon dynamic games is not straightforward, even for games with a linear system and quadratic cost functions (LQ games) \cite{siam2024, TACgame2024}. The coupled equations of FNEs typically require solving unknown high-dimensional multivariate matrices, numerous cross-product interactions among them, and high-order equations. These factors collectively contribute to the complexity \cite{2025finitestrategy}. Moreover, the LQ games may admit multiple FNEs \cite{TACgame2024,ifac2023}. Thus, approximate Nash equilibria and other iterative solutions are considered in both discrete-time setting \cite{TACgame2024,ifac2023} and continuous-time setting \cite{siopt2023,mor2024}. For example, Nortmann et al. \cite{ifac2023} considers the $\epsilon_{\alpha, \beta}$-Nash equilibrium as an approximate solution for the infinite-horizon discrete-time LQ game, which generalizes the $\epsilon_\alpha$-Nash equilibrium introduced in \cite{tac_approximate_2014}. 

Compared to the abundant results in continuous-time differential games \cite{liu_tac_diff_2023,siam_diff_2023,siam-diff_2021,auto_diff_2023}, the discrete-time scenario has received less attention \cite{TACgame2024,automatica_game2025}. However, the computation of FNEs for infinite-horizon discrete-time LQ games has been actively studied in recent years. Nortmann et al. \cite{TACgame2024} propose four iterative algorithms to obtain the feedback matrices of FNEs in infinite-horizon LQ discrete-time games, and discusses several locally asymptotically stable conditions for FNEs based on the vectorization of the coupled matrix difference equations associated with FNEs. For scalar scenarios of the game, that is, when the state, each player's input, and the cost matrices are all one-dimensional, Nortmann et al. \cite{automatica_game2025} provide a graphical representation of FNEs, showing that the intersection points between horizontal lines and the corresponding auxiliary functions correspond to FNE solutions. It also characterizes the number and properties of FNEs based on ex-ante parameters. Monti el al. \cite{siam2024} provide a necessary and sufficient condition for all FNEs of the game based on the stable solutions of the coupled generalized algebraic Riccati equations, supplementing the results in \cite{1998basarNoncooperativeGame}.
Despite these studies offering detailed characterizations of FNEs in terms of iterative solutions, stability conditions, or the number of equilibria, they are not adaptive when both the system dynamics and the state are unknown to the players.

For infinite-horizon discrete-time LQ games with unknown system dynamics, \cite{access2_2019,Cybernetics2019,access1_2019} consider several off-policy reinforcement learning (RL) algorithms to solve the game with unknown input/state (i/s) dynamics. These algorithms obtain the FNE through an iterative learning process or data collection following specific rules. 
However, many scenarios assume that the game begins as soon as each player provides her first input, and the earlier steps are more crucial than the later ones due to the discount factors, such as in the macroeconomic case discussed in \cite{kydland1976}. In such scenarios, players may have an incentive to modify their strategies during the iteration or data collection processes to achieve lower costs. 
The data-driven method for a discrete-time LQ game with unknown dynamics in \cite{TACgame2024} requires cooperation between players and assumes that they can schedule experiments and recursively collect finite-length data of the state and their own inputs. However, cooperative tests are not always suitable, and similarly, each player may adopt other strategies to increase her payoff during the scheduled experiments in the data collection period. Moreover, the studies on unknown-dynamics games focus only on unknown parameters of i/s systems, where the state is always public to the players. When the input/output/state (i/o/s) system is known, it can be equivalently represented by an i/s system. However, if the states (including the initial state) of the i/o/s system remain unknown to the players, the game becomes equivalent to one with the corresponding i/s system involving unknown system parameters and partially unknown states, which is more challenging than the classical setting with unknown i/s dynamics. 

Motivated by these problems, we introduce an $N$-player discrete-time LQ game model with an unknown i/o/s linear system and offline input/output data. Compared to the classical model in \cite{1998basarNoncooperativeGame,TACgame2024}, we provide the following characterizations:
(i) in addition to the parameters of the i/o/s system, the states are also unknown;
(ii) players are given an offline input/output dataset that only needs to satisfy  mild excitation or rank conditions, without requiring coordination or adherence to specific data collection or iteration rules; every step counts, as no trial or iteration phase is allowed once the game begins; 
(iii) players have different output reference trajectories and discount factors,  the former significantly increases competition, while the latter emphasizes the importance of the early stages of the game.
The main contributions of this paper are summarized as follows:
\begin{enumerate}
    \item We apply data-driven and behavioral systems methods to analyze the discrete -time linear-quadratic game with unknown input/output/state dynamics and offline input/output data. We establish two necessary and sufficient conditions for the existence and uniqueness of feedback Nash equilibria (FNEs) in the finite-horizon setting. The first condition proves that the finite-horizon game with unknown dynamics and its corresponding known-dynamics counterpart share the same FNEs, and provides explicit relationships between their respective FNE matrices. The second condition provides a set of coupled equations whose solution is directly related to the FNE of the game, and which are fully determined by the offline data and the parameters of the objective functions.
    \item To simplify the computation of FNEs in the finite-horizon game, we further propose a sufficient condition and the corresponding algorithm. Specifically, if the $T$ square matrices constructed from the offline data are invertible, then one FNE of the finite-horizon game can be obtained by solving $T$  sets of linear equations, where $T$ denotes the length of the game. 
    \item In the infinite-horizon unknown dynamics scenario, each player cannot accurately compute the FNE due to limited offline data. Under this setting, we prove that the finite-horizon strategy of ``watching a finite number of steps into the future and moving one step now,'' in which each player treats the infinite-horizon game with unknown dynamics as a finite-horizon game at each stage, provides an approximate solution in terms of total cost under the assumptions in \cite{2025finitestrategy}. When all players adopt this finite-horizon strategy, we prove that the total cost of each player converges to the cost under the related FNE of the infinite-horizon game with known dynamics as the players’ horizon length tends to infinity. Furthermore, we characterize the convergence rate of the total cost. We propose the corresponding algorithm and demonstrate its efficacy through a non-scalar numerical example.
\end{enumerate}

The remainder of the paper is organized as follows. Section \ref{Preliminaries} reviews the basic theory of discrete-time LQ games, LTI systems, and related fundamental lemmas. Sections \ref{finite-horizon Unknown-dynamics Game} and \ref{Infinite-horizon Unknown-dynamics Game} present our main results for the finite- and infinite-horizon games, respectively. Section \ref{Numerical Study and Further Analysis} provides a non-scalar numerical study, and Section \ref{conclusion} concludes the paper.

\section{Preliminaries}
\label{Preliminaries}
This section provides the necessary background on $N$-player, $T$-stage discrete-time dynamic game models with linear input/output/state (i/o/s) dynamics and quadratic objective functions, incorporating different discount factors and output reference trajectories for the players. It then reviews the basic theory of LTI behavioral systems and two related fundamental lemmas, which are crucial for constructing the data-driven method used to obtain the FNEs of the unknown-dynamics game.

\subsection{Linear-Quadratic Discrete-time Dynamic Game}

We consider an $N$-player $T$-stage discrete-time dynamic game characterized by the linear input/output/ state  system dynamics and quadratic objective functions below. Each player has her own discount factor and output reference trajectory. Let $\mathbb{N} = \{1, 2, \dots, N\}$ be the set of players, and let $\mathbb{T} = \{1, 2, \dots, T\}$ be the set of stages.
For the system (\ref{isosystem}), let
$$u_t=\text{col}(u_t^1,...,u_t^N)=\begin{bmatrix}
u_t^1\\
\vdots\\
u_t^N
\end{bmatrix}\in \mathbb{R}^{m \times 1},$$ where $u_t^i \in \mathbb{R}^{m_i}$ represents the input controlled by player $i$ at stage $t$, for $i \in \mathbb{N}$ and $t \in \mathbb{T}$, and $m = m_1 + \dots + m_N$.
The state and output at stage $t$ are denoted by $x_t \in \mathbb{R}^n$ and $y_t \in \mathbb{R}^p$, respectively.
The cost function of each player $i \in \mathbb{N}$ is given by (\ref{costfun}), where $l_t^i$, $t \in \mathbb{T}$, is the exogenous output reference trajectory of player $i$, $\delta_i \in (0,1]$ is the discount factor, the cost matrices $Q^i$ and $R^{ii}$ are positive semi-definite and positive definite, respectively.

\begin{align}
&\left\{\begin{aligned}
&x_{t+1} = Ax_t + Bu_t = Ax_t + \sum\limits_{i=1}^N B^iu^i_t\\
&y_t = Cx_{t} + Du_t=Cx_t + \sum\limits_{i=1}^N D^iu^i_t
\end{aligned}\right.
\label{isosystem}\\
&J^i(x_1,u)=\frac{1}{2} \sum\limits_{t=1}^{T} [(y_t-l^i_t)^{\top} Q^i (y_t-l^i_t) + \sum\limits_{j =1} ^N{(u^j_t)^{\top} R^{ij} u^j_t}] (\delta_i)^{t-1}
\label{costfun}
\end{align}

The information structure assumes that, at each stage $t \in \mathbb{T}$, every player can observe  the states $x_1, \dots, x_t$, the outputs $y_1, \dots, y_{t-1}$, and the inputs $u_1, \dots, u_{t-1}$. Moreover, all time-invariant parameters including $A$, $B$, $C$, $D$, $Q^i$, $R^{ij}$, $l^i_t$, and $\delta^i$ are known to all players for any $i, j \in \mathbb{N}$ and $t \in \mathbb{T}$.
Let $\gamma_t^i$ denote the strategy of player $i$ at stage $t$, defined as a map from the current state $x_t$ to the input $u_t^i$, that is, $\gamma_t^i(x_t) = u_t^i$. The definition of feedback Nash equilibrium (FNE) is given below. It can be readily verified that this  definition is equivalent to Definition 3.22 in \cite{1998basarNoncooperativeGame}, which is a subgame perfect Nash equilibrium.

\begin{definition}[Feedback Nash Equilibrium]
\label{de:FNE_finite}
 Consider a discrete-time T-stage game with (\ref{isosystem}, \ref{costfun}). 
 The strategy profile $\gamma^*=\{\gamma^{1*},$ $ \gamma^{2*},...,\gamma^{N*}\}$ is an FNE if and only if for any initial state $x_1 \in \mathbb R^n, \forall t\in \mathbb T, \forall i \in \mathbb N$, the following inequality  
 \begin{align*}
&J^i(x_1,\gamma_1,...,\gamma_{t-1},\gamma^{i*}_t,...,\gamma^{i*}_T,\gamma^{(-i)*}_t,...,\gamma^{(-i)*}_T) \\ 
\leq &J^i(x_1,\gamma_1,...,\gamma_{t-1},\gamma^{i}_t,...,\gamma^{i}_T,\gamma^{(-i)*}_t,...,\gamma^{(-i)*}_T), 
 \end{align*}
holds for any strategy $\gamma_1,...,\gamma_{t-1},\gamma^{i}_t,...,\gamma^{i}_T$, where $\gamma_t=\{\gamma^{1}_t,...,\gamma^{N}_t\}$, $\gamma^i=\{\gamma^{i}_1,...,\gamma^{i}_T\}$, $\gamma^{i*}=\{\gamma^{i*}_1,...,\gamma^{i*}_T\}$, $\gamma^{(-i)*}_t= \{\gamma^{1*}_t,$ $...,\gamma^{(i-1)*}_t,\gamma^{(i+1)*}_t,$ $...,\gamma^{N*}_t\}$. 
\end{definition}

\subsection{LTI System and Fundamental Lemmas}

A behavioral dynamical system is defined as a triple $\Sigma = (T, W, \mathcal{B})$, where $T$ is the time axis, $W$ is the signal space, and $\mathcal{B} \subset W^{T}$ is the set of behaviors. $\Sigma$ is said to be linear if $W$ is a linear space and $\mathcal{B}$ is a linear subspace of $W^{T}$. When $T = \mathbb{Z}$, $\Sigma$ is time-invariant if $\sigma^{t} \mathcal{B} = \mathcal{B}$, that is, $\sigma^{t} w \in \mathcal{B}$ for all $w \in \mathcal{B}$, where the shift operator $\sigma^{t}$ is defined by $(\sigma^{t} w)(t') = w(t + t')$.
Let $\mathcal{B}_t = \{ v \in W^{t} \mid \exists w \in \mathcal{B}, \text{ s.t. } v_k = w_k, 1 \leq k \leq t \}$ denote the restricted behavior of length $t$.
The set of all linear, time-invariant behavioral dynamical systems with finite-dimensional signal space $W = \mathbb{R}^q$ is denoted by $\mathcal{L}^q$, and is referred to as the class of LTI systems.

To simplify notation, we use the behavior $\mathcal{B}$ to represent the corresponding dynamical system $\Sigma = (\mathbb{Z}, \mathbb{R}^{q}, \mathcal{B})$, following \cite{DEEPC2019ECC}. Any such LTI system admits an equivalent input/output/state  representation
$$\mathcal{B}(A, B, C, D) = \{w = \text{col}(u, y) \mid \exists x \text{ such that } \sigma x = Ax + Bu,~ y = Cx + Du\},$$
as shown in \cite{1986willemspart1}. If an i/o/s  representation $\mathcal{B}(A, B, C, D)$ of $\mathcal{B}$ has the smallest possible state dimension among all such representations, it is called a minimal representation. In this case, the smallest positive integer $l$ such that $\text{col}(C, CA, \dots, CA^{l-1})$ has full column rank is denoted by the lag $\mathbf{l}(\mathcal{B})$. The input, output, and state dimensions of $\mathcal{B}$, denoted by $\mathbf{m}(\mathcal{B})$, $\mathbf{p}(\mathcal{B})$, and $\mathbf{n}(\mathcal{B})$, respectively, are defined as those of $\mathcal{B}(A, B, C, D)$. We then adopt the following notation to classify all LTI systems, as in \cite{nonpersistentexcitation2022}.
\[\partial \mathcal{L}^{q,n}_{m,l}=\{\mathcal{B} \in \mathcal{L}^q|  \mathbf{m}(\mathcal{B})=m, \mathbf{p}(\mathcal{B})=q-m, \mathbf{n}(\mathcal{B})=n, \mathbf{l}(\mathcal{B})=l\}.\]

\begin{align}
\tag{2.3}
\mathcal{H}_L(v)=\begin{bmatrix}
v_1 & v_2 & \dots & v_{s-L+1}\\
v_2 & v_3 & \dots & v_{s-L+2}\\
\vdots & \vdots & \dots & \vdots\\
v_L & v_{L+1} & \dots & v_s
\end{bmatrix}.
\label{Hankel}
\end{align}

For a time series $v = (v_1, \dots, v_s)$ with $v_i \in \mathbb{R}^r$ for $i = 1, \dots, s$, its Hankel matrix with $L$ block rows is defined as in (\ref{Hankel}). For multiple time series $v^j$, $j = 1, \dots, t$, possibly of different lengths, the corresponding Hankel matrix is given by $\mathcal{H}_L(v^1,v^2,...,v^{t})=[\mathcal{H}_L(v^1),\dots, \mathcal{H}_L(v^{t})]$, as introduced in \cite{nonpersistentexcitation2022}.

\begin{lemma}[\cite{initialcondition2008}, Lemma 1]
    Given a LTI system $\{Z_{\geq 0}, R^{m+p}, \mathcal{B}\}$ $\in \partial \mathcal{L}^{q,n}_{m,l}$ with an i/o/s representation $\mathcal{B}(A,B,C,D)$(not necessarily minimal). If the $T_{\text{ini}}$-length trajectory
    $ w_{1\leq t\leq T_{\text{ini}}} \in \mathcal{B}_{T_{\text{ini}}}$ is given, $T_{\text{ini}}\geq \mathbf{l}(A,C)$, then $x_1,$ $...,x_{T_{\text{ini}}+1}$ are uniquely determined.
\label{initialcondition2008}
\end{lemma}

Lemma \ref{initialcondition2008} is derived by modifying Lemma 1 in \cite{initialcondition2008}, and its proof follows a similar argument. It implies that a historical input/output trajectory can uniquely determine the states of the system during that trajectory, and the result still holds even if the minimality of the i/o/s system is unknown.

\begin{lemma}[\cite{nonpersistentexcitation2022}, Corollary 19]
Let the offline data $\mathcal{W}_d = (w_d^1, w_d^2, \dots, w_d^v)$ be generated by an LTI system $\mathcal{B} \in \partial \mathcal{L}^{q,n}_{m,l}$, where $w_d^i \in \mathcal{B}_{T_i}$ is a trajectory with length $T_i$.
Then, the image of $\mathcal{H}_L(\mathcal{W}_d)$ is $\mathcal{B}_L$ for $L > \mathbf{l}(\mathcal{B})$ if and only if $\text{rank}(\mathcal{H}_L(\mathcal{W}_d)) = mL + n$. Furthermore, if $\{L > \mathbf{l}(\mathcal{B}) \mid \text{rank}(\mathcal{H}_L(\mathcal{W}_d)) = mL + n\} \neq \emptyset$ with its maximum element $L_m$, then it equals $\{\mathbf{l}(\mathcal{B}) + 1, \mathbf{l}(\mathcal{B}) + 2, \dots, L_m\}$.
\label{nonpersistentexcitation2022}
\end{lemma}

Lemma \ref{nonpersistentexcitation2022} characterizes which offline data can linearly represent all possible input and output trajectories within a specified time horizon \cite{nonpersistentexcitation2022}. To apply this lemma for prediction, suppose that $w_t = \mathrm{col}(u_t, y_t)$, where $u^d$ and $y^d$ denote the input and output components of $\mathcal{W}^d$, respectively. Let $L = T_{\mathrm{ini}} + T$, and assume that the initial $T_{\mathrm{ini}}$ steps $w_t$ for $t = -T_{\mathrm{ini}} + 1, \ldots, 0$ are given. Then, all possible future trajectories $w_t$ for $t = 1, 2, \ldots, T$ satisfy 
\[\{(u, y) \in \mathbb{R}^{Tm} \times \mathbb{R}^{Tp}\mid\exists g ~ s.t.~\begin{bmatrix}
 U_p \\
 Y_p \\
 U_f \\
 Y_f
 \end{bmatrix}g=
 \begin{bmatrix}
 u_{\text{ini}} \\
 y_{\text{ini}} \\
 u \\
 y
 \end{bmatrix}\},\] where
\[
\begin{bmatrix} 
U_{\mathrm{p}} \\ 
U_{\mathrm{f}} 
\end{bmatrix} = \mathcal{H}_{T_{\mathrm{ini}} + T}(u^{\mathrm{d}}), \quad
\begin{bmatrix} 
Y_{\mathrm{p}} \\ 
Y_{\mathrm{f}} 
\end{bmatrix} = \mathcal{H}_{T_{\mathrm{ini}} + T}(y^{\mathrm{d}}),
\]
with \( U_{\mathrm{p}} \) denoting the first \( T_{\mathrm{ini}} \) block rows of \(\mathcal{H}_{T_{\mathrm{ini}} + T}(u^{\mathrm{d}})\), and \( U_{\mathrm{f}} \) denoting the last \( T \) block rows of \(\mathcal{H}_{T_{\mathrm{ini}} + T}(u^{\mathrm{d}})\).  \( Y_{\mathrm{p}} \) and \( Y_{\mathrm{f}} \) are similarly defined \cite{DEEPC2019ECC}. 
For an i/o/s representation \(\mathcal{B}(A,B,C,D)\), when
$
T_{\mathrm{ini}} \geq \mathbf{l}(A,C),$ 
Lemma \ref{initialcondition2008} ensures that the initial state \( x_1 \) is uniquely determined by the given initial inputs and outputs \( u_{\mathrm{ini}}, y_{\mathrm{ini}} \). Then, from Lemma \ref{nonpersistentexcitation2022}, for any \( u \in \mathbb{R}^m \), since all possible future trajectories can be represented, there exists a vector \( g \) such that the relevant equations hold, and a unique output \( y \) satisfies these equations due to the uniqueness of \(x_1 \) \cite{DEEPC2019ECC}.

\section{Finite-horizon Unknown-dynamics Game}
\label{finite-horizon Unknown-dynamics Game}

In the unknown-dynamics game with the linear system (\ref{isosystem}), the matrices $A, B, C, D$ and the state $x_t$ (including the initial state $x_1$) are unknown to all players in $\mathbb{N}$. The corresponding behavioral system is denoted by $\mathcal{B}$, and the players' objective functions are given by (\ref{costfun}). The game spans $T$ stages, with $\mathbb{T} = {1, 2, \dots, T}$. At stage $t \in \mathbb{T}$, each player has access to the input-output history $\{u^j_k, y_k \mid k \leq t-1, j \in \mathbb{N}\}$, along with the initial data and a set of offline data trajectories. The initial data set consists of $u_{\text{ini}} = \text{col}(u_{-T_{\text{ini}}+1}, \dots, u_0)$ and $y_{\text{ini}} = \text{col}(y_{-T_{\text{ini}}+1}, \dots, y_0)$, where $T_{\text{ini}}$ is an exogenous positive integer.  The offline data set comprises $S$ trajectories of $u$ and $y$ with varying lengths, denoted by $w^s_d \in \mathcal{B}_{T_s}$ for $s = 1, \dots, S$, and collected as $\mathcal{W}_d = \{w^1_d, \dots, w^S_d\}$. The decision variable of player $i \in \mathbb{N}$ is $u^i_t \in \mathbb{R}^{m_i}$, and her strategy is a mapping from the initial data and past inputs to her decision, i.e., $u^i_t = \gamma(u_{\text{ini}}, y_{\text{ini}}, u_1, \dots, u_{t-1})$. For compactness, we define $U_t = \text{col}(u_{\text{ini}}, y_{\text{ini}}, u_1, \dots, u_t)$ for $t = 0, 1, \dots, T-1$. We now introduce a key assumption on the data.

\begin{assumption}
1. Initial data: $T_{\text{ini}} \geq \mathbf{l}(A,C)$.
2. Offline data: $\text{rank}(\mathcal{H}_L(\mathcal{W}_d))$ $=\mathbf{m}(\mathcal{B})L+\mathbf{n}(\mathcal{B})$, $L> T_{\text{ini}}+ T$.
\label{A1}
\end{assumption}

Next, we present our main results for the finite-horizon unknown-dynamics game. We begin by defining the feedback Nash equilibrium in this setting. We provide two conditions that ensure the existence and uniqueness of the FNE. These results include the affine structure of the equilibrium, detailed relations between the FNEs in the known- and unknown-dynamics games, and the coupled equations characterizing the FNE based on offline data. Finally, we provide a sufficient condition that simplifies the computation of the FNE, together with a corresponding algorithm.   

\begin{definition}[FNE in finite-horizon unknown-dynamics games]
Suppose that Assumption \ref{A1} holds for the T-stage unknown-dynamics game (\ref{isosystem}, \ref{costfun}).   The strategy profile $\gamma^*=\{\gamma^{1*},$ $ \gamma^{2*},...,\gamma^{N*}\}$ is an FNE if and only if for any given initial data $\text{col}(u_{\text{ini}},y_{\text{ini}}) \in \mathcal{B}_{T_{\text{ini}}}$,  $ \forall t\in \mathbb T, \forall i \in \mathbb N$, the following inequality  
 \begin{align*}
&J^i(u_{\text{ini}},y_{\text{ini}},\gamma_1,...,\gamma_{t-1},\gamma^{i*}_t,...,\gamma^{i*}_T,\gamma^{(-i)*}_t,...,\gamma^{(-i)*}_T) \\ 
\leq &J^i(u_{\text{ini}},y_{\text{ini}},\gamma_1,...,\gamma_{t-1},\gamma^{i}_t,...,\gamma^{i}_T,\gamma^{(-i)*}_t,...,\gamma^{(-i)*}_T), 
 \end{align*}
holds for any strategy $\gamma_1,...,\gamma_{t-1},\gamma^{i}_t,...,\gamma^{i}_T$, where $\gamma_t=\{\gamma^{1}_t,...,\gamma^{N}_t\}$, $\gamma^i=\{\gamma^{i}_1,...,\gamma^{i}_T\}$, $\gamma^{i*}=\{\gamma^{i*}_1,...,\gamma^{i*}_T\}$, $\gamma^{(-i)*}_t= \{\gamma^{1*}_t,$ $...,\gamma^{(i-1)*}_t,\gamma^{(i+1)*}_t,$ $...,\gamma^{N*}_t\}$, $J^i(u_{\text{ini}},y_{\text{ini}},\gamma)$ denotes the cost for player $i$ under the initial data $u_{\text{ini}},y_{\text{ini}}$ and the strategy profile $\gamma$.
\end{definition}

\begin{theorem}
     Consider an unknown-dynamics T-stage game (\ref{isosystem}, \ref{costfun}) with offline data $\mathcal{W}_d$ and initial data $u_{\text{ini}}, y_{\text{ini}}$ generated from an underlying dynamics (\ref{isosystem}). Suppose that Assumption \ref{A1} holds. Then, (i) the T-stage unknown-dynamics game and the corresponding known-dynamics game have the same set of FNEs, including the case when such a set is empty, (ii) for any FNE of the unknown-dynamics game, each player's strategy  takes an affine form: 
    $u^{i*}_t = \gamma^{i*}_t(U_{t-1}) = K^{i*}_t U_{t-1} + L^{i*}_t$, 
    and her utility takes a quadratic form: 
    $J^{i*} = \frac{1}{2} U_{0}^{\top} P_1^{i*} U_{0} + (S^{i*}_1)^{\top} U_{0} + w^{i*}_{1},
 t\in \mathbb T, i\in \mathbb N,$  
where $ P^{i*}_t, S^{i*}_t, w^{i*}_t, K^{i*}_t, L^{i*}_t, t \in \mathbb T, i\in \mathbb N$ are constant matrices determined by offline data $\mathcal{W}_d$ and $Q^i, R^{ij}, \delta_i, l^i_t$. 
    \label{theorem 1}
\end{theorem}

\begin{proof}
    See Appendix.\ref{Proof of Theorem 1}. 
\end{proof}

Theorem \ref{theorem 1} shows the equivalence between the FNEs of games with known and unknown dynamics. Note that this theorem assumes the initial data is given. According to Lemma \ref{initialcondition2008} and Assumption \ref{A1}, the underlying initial state $x_1$ is uniquely determined by $u_{\text{ini}}, y_{\text{ini}}$, and we refer to $x_1$ as the initial state matched with $(u_{\text{ini}}, y_{\text{ini}})$. Therefore, the FNE strategies here are intrinsically a set of maps from $(u_1, \ldots, u_{t-1})$ to $u^i_t$, for $i \in \mathbb{N}$ and $t \in \mathbb{T}$, in both known- and unknown-dynamics games. For any FNE here, on the one hand, it has the form $u^{i*}_t = \overline{K}^{i*}_t x_t + \overline{L}^{i*}_t$ \cite{1998basarNoncooperativeGame}. On the other hand, it has the form $u^{i*}_t = K^{i*}_t U_{t-1} + L^{i*}_t$ by Theorem \ref{theorem 1}. Thus, the equivalence of the FNEs in Theorem \ref{theorem 1} implies a functional relationship involving $u_1, \ldots, u_{t-1}$, as represented by the following equation. 
\begin{equation}
\begin{aligned}
    & \overline K^{i*}_t x_t+ \overline L^{i*}_t =\overline K^{i*}_t(A^{t-1}x_1 + A^{t-2}Bu_1 +...+ ABu_{t-2} + Bu_{t-1}) + \overline L^{i*}_t\\
    =&K^{i*}_t U_{t-1}+ L^{i*}_t= K^{i*}_t \text{col}(u_{\text{ini}},y_{\text{ini}},u_1,...,u_{t-1})+ L^{i*}_t, \forall u_1,...,u_{t-1} \in\mathbb R^m. 
    \end{aligned}
    \label{FNE equivalence equa}
    \end{equation} 
    
To further clarify this equivalence, we first introduce some necessary notation.

{\bf Notation}  
For $x=\text{col}(x_1,...,x_s)=(x_1^{\top},...,x_s^{\top})^{\top},$ where $x_i \in \mathbb{R}^{m_i \times 1}$ and $A\in \mathbb{R}^{g \times m},$ with $m=m_1+...+m_s$, $A_{:,x_i} \in \mathbb{R}^{g \times m_i}$ denotes the block from the $(m_1+...+m_{i-1}+1)$-th column to the $(m_1+...+m_{i})$-th column of $A$. By definition, we have $$ Ax=[A_{:,x_1},...,A_{:,x_s}]\text{col}(x_1,...,x_s)=A_{:,x_1} x_1+...+ A_{:,x_s}x_s.$$
For $y=(y_1,...,y_s),$ where $ y_i \in \mathbb{R}^{1 \times m_i}$ and $B\in \mathbb{R}^{m \times g},$ with 
$ m=m_1+...+m_s$, $B_{y_i,:}\in \mathbb{R}^{m_i \times g}$ denotes the block from the $(m_1+...+m_{i-1}+1)$-th row to the $(m_1+...+m_{i})$-th row of $B$. 
Similarly, we have $$ yB=(y_1,...y_s)\text{col}[B_{y_1,:},...,B_{y_s,:}]=y_1 B_{y_1,:} +...+ y_s B_{y_s,:}.$$ 
Using the notations above, we have
$$K^{i*}_t=[(K^{i*}_t)_{:,(u_{\text{ini}},y_{\text{ini}})}, (K^{i*}_t)_{:,u_{1}},...,(K^{i*}_t)_{:,u_{t-1}}].$$ 
Given the initial data $u_{\text{ini}}, y_{\text{ini}}$ and the underlying state $x_1$ matched to it, since (\ref{FNE equivalence equa}) holds for any $u_1, \ldots, u_{t-1} \in \mathbb{R}^m$, we obtain 
\begin{equation}
\begin{aligned}
    &\overline K^{i*}_t A^{t-k-1}B = (K^{i*}_t)_{:,u_{k}},  k=1,...,t-1,\\
    &\overline K^{i*}_t A^{t-1}x_1 +\overline L^{i*}_t= (K^{i*}_t)_{:,u_{\text{ini}}} u_{\text{ini}}+(K^{i*}_t)_{:,y_{\text{ini}}} y_{\text{ini}}+L^{i*}_t,t\in \mathbb T, i\in \mathbb N.
    \end{aligned}
    \label{strategy matrices equation}
    \end{equation}

\begin{remark}
If the $T$-stage known-dynamics game (\ref{isosystem}, \ref{costfun}) has a unique FNE $\{\overline K^{i*}_t x_t + \overline L^{i*}_t \mid i \in \mathbb{N}, t \in \mathbb{T}\}$, then (\ref{strategy matrices equation}) can be expressed in a more refined form. This is because, for an FNE of the corresponding unknown-dynamics game $\{ K^{i*}_t U_{t-1} + L^{i*}_t \mid i \in \mathbb{N}, t \in \mathbb{T}\}$ and for any initial data $(u_{\text{ini}}, y_{\text{ini}}) \in \mathcal{B}_{T_{\text{ini}}}$, (\ref{strategy matrices equation}) always holds due to the uniqueness of the FNE of the known-dynamics game. Note that $x_1=0$ is matched with $u_{\text{ini}}, y_{\text{ini}} = \mathbf{0}$, then we have  
\begin{equation}
\begin{aligned}
    &\overline K^{i*}_t A^{t-k-1}B = (K^{i*}_t)_{:,u_{k}},  k=1,...,t-1,\\
    &\overline K^{i*}_t A^{t-1}x_1 = (K^{i*}_t)_{:,(u_{\text{ini}},y_{\text{ini}})} \text{col}(u_{\text{ini}},y_{\text{ini}})=(K^{i*}_t)_{:,u_{\text{ini}}} u_{\text{ini}}+(K^{i*}_t)_{:,y_{\text{ini}}} y_{\text{ini}}, \\
    &\overline L^{i*}_t = L^{i*}_t, t\in \mathbb T, i\in \mathbb N.
    \end{aligned}
    \label{strategy matrices equation2}
    \end{equation} 
 This implies that for the unique FNE of the known-dynamics game, there exists an FNE of the  unknown-dynamics game satisfying (\ref{strategy matrices equation2}) for every possible initial data $\operatorname{col}(u_{\text{ini}}, y_{\text{ini}}) \in \mathcal{B}_{T{\text{ini}}}$. 
 \label{remark}
\end{remark}

\begin{remark}
For the $T$-stage unknown-dynamics game (\ref{isosystem}, \ref{costfun}), given two distinct sets of strategy matrices
$\{\hat K^{i}_t, \hat L^{i}_t \mid t\in \mathbb{T}, i\in \mathbb{N}\}$ and
$\{K^{i}_t, L^{i}_t \mid t\in \mathbb{T}, i\in \mathbb{N}\}$,
the strategies they form for player $i$ are considered identical if the following condition holds for any $ (u_{\text{ini}},y_{\text{ini}}) \in \mathcal{B}_{T_{\text{ini}}}, (u_1,...,u_{t-1})\in \mathbb{R}^{m(t-1)},$  
\[\hat K^i_t U_{t-1} + \hat L^i_t = K^{i}_t U_{t-1} + L^{i}_t,~ t\in \mathbb T. \]
\end{remark}

\begin{remark}
Theorem \ref{theorem 1} implies that if player $i$ does not know the system dynamics but player $j$ does, player $i$ can perform just as well as if she knew the dynamics, provided she has sufficient offline data. Therefore,  information about the system dynamics does not provide an advantage to player $j$.
\end{remark}

Based on structured $N$-person data-enabled Bellman equations constructed from offline data, we derive the following coupled equations for all $t \in \mathbb{T}$ and $i \in \mathbb{N}$, which provide necessary and sufficient conditions for the existence and uniqueness of the FNE in the game. 
\begin{align}
    &\begin{aligned}
 &[(G_t)_{:,u^i_t}]^{\top} Q^i G_t(\mathbf{K}_t)_{:, (u_{\text{ini}},y_{\text{ini}})}\textbf B(\mathcal{B}_{ini}) 
 + \delta_i(P^{i}_{t+1})_{u^i_t,:}(\mathbf{K}_t)_{:, (u_{\text{ini}},y_{\text{ini}})}\textbf B(\mathcal{B}_{ini}) \\
 &+ R^{ii}(K^i_t)_{:, (u_{\text{ini}},y_{\text{ini}})}\textbf B(\mathcal{B}_{ini}) =0,
\end{aligned}
\tag{1-1a}
\label{eq1-1a}\\
    &\begin{aligned}
 &[(G_t)_{:,u^i_t}]^{\top} Q^i G_t (\mathbf{K}_t)_{:,u_j}
 + R^{ii}(K_t^{i})_{:,u_j} 
 + \delta_i (P^{*}_{t+1})_{u^i_t,:}(\mathbf{K}_t)_{:,u_j} =0,\\
 &j=1,...,t-1,
\end{aligned}
\tag{1-1b}
\label{eq1-1b}\\
&\begin{aligned}
&[(G_t)_{:,u^i_t}]^{\top} Q^i (G_t \mathbf{L}_t - l^i_t) + R^{ii} L_t^i + \delta_i (P^i_{t+1})_{u^i_t,:}\mathbf{L}_t + \delta_i (S^i_{t+1})_{u^i_t,:} = 0, 
\end{aligned}
\tag{1-2}
\label{eq1-2}\\
&\begin{aligned}
&P^i_t = (G_t \mathbf{K}_t)^{\top} Q^i (G_t \mathbf{K}_t) + \sum\limits_{j\in N}(K^j_t)^{\top} R^{ij} (K^j_t) + \delta_i(\mathbf{K}_t)^{\top} P_{t+1}^i \mathbf{K}_t,\\
&~\text{with}~P_{T+1}^i =0, 
\end{aligned}
\tag{1-3}
\label{eq1-3}\\
&\begin{aligned}
S^i_t &= (G_t \mathbf{K}_t)^{\top} Q^i (G_t \mathbf{L}_t - l^i_t) + \sum\limits_{j\in N}(K^j_t)^{\top} R^{ij} (L^j_t) + \delta_i(\mathbf{K}_t)^{\top} P_{t+1}^i \mathbf{L}_t \\
&+ \delta_i\mathbf{K}_t^{\top} S_{t+1}^i,
~\text{with}~S_{T+1}^i =0, 
\end{aligned}
\tag{1-4}
\label{eq1-4}\\
&\begin{aligned}
w^i_t &= \frac{1}{2}(G_t \mathbf{L}_t - l^i_t)^{\top} Q^i (G_t \mathbf{L}_t - l^i_t) + \frac{1}{2}\sum\limits_{j\in N}(L^j_t)^{\top} R^{ij} (L^j_t)  \\
&+ \frac{1}{2}\delta_i(\mathbf{L}_t)^{\top} P_{t+1}^i \mathbf{L}_t 
+ \delta_i\mathbf{L}_t^{\top} S_{t+1}^i + \delta_iw_{t+1}^i, ~\text{with}~ w_{T+1}^i =0.
\end{aligned}
\tag{1-5}
\label{eq1-5}
\end{align}

\begin{theorem}
\label{theorem 2} Consider an unknown-dynamics T-stage game (\ref{isosystem}, \ref{costfun}) with offline data $\mathcal{W}_d$ and initial data $u_{\text{ini}}, y_{\text{ini}}$ generated from an underlying dynamics (\ref{isosystem}).
Suppose that Assumption \ref{A1} holds. Then, (I) The game has at least one FNE if and only if   (\ref{eq1-1a}, \ref{eq1-1b}, \ref{eq1-2} - \ref{eq1-5}) have at least one solution.
Specifically, (a) for an FNE $\{u^{i*}_t=K^{i*}_t  U_{t-1} + L^{i*}_t\mid t\in \mathbb T, i\in \mathbb N\}$, there must exist $\{P^{i*}_t, S^{i*}_t, w^{i*}_t\mid t\in \mathbb T, i\in \mathbb N\}$ such that $\{P^{i*}_t, S^{i*}_t, w^{i*}_t, K^{i*}_t, L^{i*}_t\mid t\in \mathbb T, i\in \mathbb N\}$ is a solution of equations (\ref{eq1-1a}, \ref{eq1-1b}, \ref{eq1-2} - \ref{eq1-5}); (b) for a solution of (\ref{eq1-1a}, \ref{eq1-1b}, \ref{eq1-2} - \ref{eq1-5}) $\{P^{i*}_t, S^{i*}_t, w^{i*}_t, K^{i*}_t, L^{i*}_t\mid t\in \mathbb T, i\in \mathbb N\}$,  $\{u^{i*}_t=K^{i*}_t  U_{t-1} + L^{i*}_t\mid t\in \mathbb T, i\in \mathbb N\}$ is an FNE of the game. In both cases, the cost of each player $i$ takes the quadratic form $$J^{i*} = \frac{1}{2} U_{0}^{\top} P_1^{i*} U_{0} + (S^{i*}_1)^{\top} U_{0} + w^{i*}_{1}, i\in \mathbb N.$$
(II) The game (\ref{isosystem}, \ref{costfun}) has a unique FNE if and only if   (\ref{eq1-1a}, \ref{eq1-1b}, \ref{eq1-2} - \ref{eq1-5}) have at least one solution,
and all solutions of (\ref{eq1-1a}, \ref{eq1-1b}, \ref{eq1-2} - \ref{eq1-5}) satisfy the following conditions: 
\begin{enumerate}
    \item They have the same constant parts $L^{i*}_t$  for any $t\in \mathbb T$ and $i\in \mathbb N.$
    \item They have the same linear parts in the corresponding columns of $u_1,u_2,$ $...,u_{t-1}$: $(K^{i*}_t)_{:,u_1},(K^{i*}_t)_{:,u_2},...,(K^{i*}_t)_{:,u_{t-1}}$, for any $t\in \mathbb T$ and $i\in \mathbb N.$ 
    \item The following condition always holds for any $\text{col}(\tilde{u}, \tilde{y}) \in \mathcal{B}_{T_{\text{ini}}}$
    \[\begin{bmatrix}(K^{i*}_t)_{:,u_{\text{ini}}}, (K^{i*}_t)_{:,y_{\text{ini}}} \end{bmatrix} \begin{bmatrix}\tilde{u}\\ \tilde{y} \end{bmatrix}=\begin{bmatrix}(\hat K^{i*}_t)_{:,u_{\text{ini}}}, (\hat  K^{i*}_t)_{:,y_{\text{ini}}} \end{bmatrix} \begin{bmatrix}\tilde{u}\\ \tilde{y} \end{bmatrix}\]
    where
    $\{P^{i*}_t, S^{i*}_t, w^{i*}_t, K^{i*}_t, L^{i*}_t\mid t\in \mathbb T, i\in \mathbb N\}$, 
    $\{\hat{P}^{i*}_t, \hat{S}^{i*}_t, \hat{w}^{i*}_t, \hat{K}^{i*}_t, \hat{L}^{i*}_t\mid t\in \mathbb T, i\in \mathbb N\} $ are any two distinct solutions of (\ref{eq1-1a}, \ref{eq1-1b}, \ref{eq1-2} - \ref{eq1-5}).
\end{enumerate}
\end{theorem}

\begin{proof}
    See Appendix.\ref{Proof of Theorem 2}. 
\end{proof}

The coupled iterative equations (\ref{eq1-1a}, \ref{eq1-1b}, \ref{eq1-2} – \ref{eq1-5}) exhibit a clear structure, as all components involving offline data depend only on $G_1, \dots, G_T$ (see the notation below). Moreover, (\ref{eq1-3} – \ref{eq1-5}) correspond to simple iterative updates of $P^{i*}_t$, $S^{i*}_t$, and $w^{i*}_t$. To clarify these equations, we first introduce the necessary notations and specify the dimensions of the relevant variables. 

{\bf Notation}  
The data matrix is $\text{col}(U_p,$ $Y_P,U_f,Y_f)$, where $\begin{bmatrix} U_{\text{p}}\\ U_{\text{f}} \end{bmatrix} :=\mathcal{H}_{T_{\text{ini}}+N}(u^{\text{d}}), $ $\begin{bmatrix}  Y_{\text{p}}\\ Y_{\text{f}} \end{bmatrix}:=\mathcal{H}_{T_{\text{ini}}+N}(y^{\text{d}})$. $U_p$ is the first $T_{\text{ini}}$ block rows in $\mathcal{H}_{T_{\text{ini}}+N}(u^{\text{d}})$, $U_f$ is the last $T$ block rows in $\mathcal{H}_{T_{\text{ini}}+N}(u^{\text{d}})$, where one block row in $U_p, U_{f}$ consists of $m$ rows, $u_t \in \mathbb R^{m}$ (similarly, $Y_p, Y_f$ are defined) \cite{DEEPC2019ECC}. $U_{ft}$ is the $t$-th block row in $U_{f}$, $(U_f)_{1:t}$ represents the first $t$ block rows in $U_{f}$. $Y_{ft}$ is the $t$-th block row in $Y_{f}$ (one block row in $Y_{f}$ consists of $p$ rows, $y_t \in \mathbb R^{p}$), $(Y_f)_{1:t}$ represents the first $t$ block rows in $Y_{f}$. $M_t = \begin{bmatrix}
U_p \\
Y_P \\
(U_f)_{1:t} 
\end{bmatrix}$, $\dag$ denotes the pseudo-inverse. $\mathbf{K}_t = \begin{bmatrix}
I \\
K_t
\end{bmatrix}, \mathbf{L}_t = \begin{bmatrix}
0 \\
L_t
\end{bmatrix},$ where $K_t= \text{col}(K_t^1,...,K_t^N), L_t= \text{col}(L_t^1,...,L_t^N),$ and $I$ is an identity matrix. It can be verified that $ U_{t} =  \begin{bmatrix}
U_{t-1}\\
u_t
\end{bmatrix}  = \begin{bmatrix}
U_{t-1}\\
K_t U_{t-1}
\end{bmatrix} + \mathbf{L}_t = \mathbf{K}_t U_{t-1} + \mathbf{L}_t$. We denote $$G_t=Y_{ft}M_t^{\dag}, t=1,...,T,$$
which is entirely determined by offline data. The matrix $\textbf B (\mathcal{B}_{T_{\text{ini}}})$ consists of a set of column vectors that form a basis for the linear space $\mathcal{B}_{T_{\text{ini}}}$, with all inputs 
$u$ placed at the top. 

{\bf Dimension} $u^i_t=K^i_t U_{t-1}+L^i_t$, $U_{t-1} \in \mathbb{R}^{[T_{\text{ini}}(m+p)+(t-1)m] \times 1}, u^i_t \in \mathbb{R}^{m_i \times 1}, K^i_t \in \mathbb{R}^{m_i \times [T_{\text{ini}}(m+p)+(t-1)m]}, L^i_t \in \mathbb{R}^{m_i \times 1}$. $P^i_{t} \in \mathbb{R}^{[T_{\text{ini}}(m+p)+(t-1)m] \times [T_{\text{ini}}(m+p)+(t-1)m]},$ $S^i_{t} \in \mathbb{R}^{[T_{\text{ini}}(m+p)+(t-1)m] \times 1},$ $w^i_{t} \in \mathbb{R}^{1 \times 1}$. $ \mathbf{K}_t \in \mathbb{R}^{[T_{\text{ini}}(m+p) + t  m] \times [T_{\text{ini}}(m+p)+(t-1)m]} $, $G_t \in \mathbb{R}^{p \times [T_{\text{ini}}(m+p)+t m]}$. Note that the number of columns in the feedback matrix $ K^i_t $ varies with the stage $ t $, which is in contrast to the known-dynamics game scenario, where the feedback matrix has a fixed number of $ n $ columns.  
\begin{align}
&[(G_t)_{:,u^i_t}]^{\top} Q^i G_t \mathbf{K}_t + R^{ii} K_t^i + \delta_i(P^i_{t+1})_{u^i_t,:}\mathbf{K}_t = 0, ~t\in \mathbb T, ~i\in \mathbb N.
\label{eq1-1}
\tag{1-1}
\end{align}

Theorem \ref{theorem 2} is not convenient for computation, as equation (\ref{eq1-1a}) takes the form $\hat A X \hat B = \hat C$, which makes computation more difficult. Here, we provide a new set of equations (\ref{eq1-1} – \ref{eq1-5}), derived as a special case of (\ref{eq1-1a}), and then combine it with (\ref{eq1-1b}) to obtain (\ref{eq1-1}), which takes the form $\hat A X = \hat b$. 
For any $t \in \mathbb{T}$, when $P_{t+1}$ is given, equation (\ref{eq1-1}) is linear in $K_t$. By combining the equations (\ref{eq1-1}) for all players at stage $t$ into a single equation, we obtain 
\begin{align}
\tilde{H}_t(P_{t+1}) K_{t} = \tilde{g}_t(P_{t+1}),
\label{eq2-1-1}
\end{align}
where $P_{t+1} = \text{col}(P^1_{t+1}, \dots, P^N_{t+1})$, $K_t = \text{col}(K^1_t, \dots, K^N_t)$.  
Similarly, for any $t \in \mathbb{T}$, when $P_{t+1}$ and $S_{t+1}$ are given, (\ref{eq1-2}) is linear in $L^i_t$. Interestingly, $L^i_t$ and $K^i_t$ share the same coefficient matrix  $\tilde{H}_t(P_{t+1})$. Thus, by combining all players' equations (\ref{eq1-2}) at stage $t$ into a single equation, we obtain  
\begin{equation}
\tilde{H}_t(P_{t+1}) L_{t} = \tilde{g}'_t(S_{t+1}),
\label{eq2-2-1}
\end{equation}
where $S_{t+1}= \text{col}(S^1_{t+1},...,S^N_{t+1})$, $L_{t}= \text{col}(L^1_{t},$$...,L^N_{t})$. The matrices $\tilde{H}_t(P_{t+1}) \in \mathbb{R}^{m \times m}$ and $\tilde{g}_t(P_{t+1}) \in \mathbb{R}^{m \times [T_{\text{ini}}(m+p) + (t-1)m]}$, $ \tilde{g}'_t(P_{t+1}) \in \mathbb{R}^{m\times 1}$ are given by the following equations. 
\begin{align*}
&\tilde{H}_t(P_{t+1}) = 
\begin{bmatrix}
(G_t^{\top} Q^1 G_t)_{u^1_t,u_t}\\
\vdots
\\
(G_t^{\top} Q^N G_t)_{u^N_t,u_t}
\end{bmatrix} + \begin{bmatrix}
\delta_1 (P^1_{t+1})_{u^1_t,u_t}\\
\vdots
\\
\delta_N (P^N_{t+1})_{u^N_t,u_t}
\end{bmatrix} + \begin{bmatrix}
R^{11} & \dots & 0 \\
\vdots & \ddots & \vdots \\
0 & \dots & R^{NN}
\end{bmatrix}
\end{align*}
\begin{align*}\tilde{g}_t(P_{t+1}) = -
\begin{bmatrix}
(G_t^{\top} Q^1 G_t)_{u^1_t,I}\\
\vdots
\\
(G_t^{\top} Q^N G_t)_{u^N_t,I}
\end{bmatrix} - \begin{bmatrix}
\delta_1 (P^1_{t+1})_{u^1_t,I}\\
\vdots
\\
\delta_N (P^N_{t+1})_{u^N_t,I}
\end{bmatrix}
\end{align*}
\begin{align*}
\tilde{g}'_t(S_{t+1}) = 
\begin{bmatrix}
(G_t^{\top} Q^1)_{u^1_t,:} l^1_t\\
\vdots
\\
(G_t^{\top} Q^N)_{u^N_t,:} l^N_t\\
\end{bmatrix} - \begin{bmatrix}
\delta_1 (S^1_{t+1})_{u^1_t,:}\\
\vdots
\\
\delta_N (S^N_{t+1})_{u^N_t,:}
\end{bmatrix}.
\end{align*}

Based on equations (\ref{eq1-1} – \ref{eq1-5}), we provide an invertibility condition for the existence of FNEs in the $T$-stage unknown-dynamics game (\ref{isosystem}, \ref{costfun}). This 
condition also enables a straightforward computation of an FNE using offline data, without requiring system identification.  Specifically, $P^{i}_{t+1}$, $S^{i}_{t+1}$, and $w^{i}_{t+1}$ are updated iteratively via equations (\ref{eq1-3} – \ref{eq1-5}), without the need to solve any additional equations. Subsequently, $K_t$ and $L_t$ are uniquely determined by solving the linear equations (\ref{eq2-1-1}) and (\ref{eq2-2-1}), guaranteed by the invertibility condition. This invertibility condition plays a central role throughout the paper. The iterative solution procedure is outlined below, and the corresponding results are formally stated in the following algorithm and theorem. 

\newpage
\textbf{Simple Iterative Process for Solving (\ref{eq1-1} - \ref{eq1-5}):}  
\begin{enumerate}
\item Initialize: set $P^{i*}_{T+1} = S^{i*}_{T+1} = w^{i*}_{T+1} = 0$.
\item Solve the linear equations (\ref{eq2-1-1}) and (\ref{eq2-2-1}) for $K^{i*}_T$ and $L^{i*}_T$.
\item Obtain $P^{i*}_T$, $S^{i*}_T$,  $w^{i*}_T$ by simple iteration of (\ref{eq1-3} – \ref{eq1-5}).
\item Solve the linear equations (\ref{eq2-1-1}) and (\ref{eq2-2-1}) for $K^{i*}_{T-1}$ and $L^{i*}_{T-1}$.
\item \dots
\item Obtain $P^{i}_1$, $S^{i}_1$,  $w^{i*}_1$ by simple iteration of (\ref{eq1-3} – \ref{eq1-5}).
\end{enumerate}
\begin{algorithm}

\caption{Backward Algorithm for T-stage Unknown-dynamics Game}\label{alg:Algorithm1}
\begin{algorithmic}
\STATE{\textbf{Input:} Offline data $\mathcal{W}_d$, initial data \(u_{\text{ini}}, y_{\text{ini}}\), length $T$, $Q^i, \delta_i, R^{ij}$, \\
$t=T$, $P^{i*}_{T+1}=0, S^{i*}_{T+1}=0, w^{i*}_{T+1}=0$, $i,j \in \mathbb{N}$}

\WHILE{\text{$|\tilde{H}_t(P^*_{t+1})| \neq 0$ and $t > 0$}} 

    \STATE{ $K_t^* \gets \tilde{H}_t(P^*_{t+1})^{-1}\tilde{g}_t(P^*_{t+1}), \quad L_t^* \gets \tilde{H}_t(P^*_{t+1})^{-1}\tilde{g}'_t(S^*_{t+1})$}
    \STATE{ Update $P^{i*}_t, S^{i*}_t, w^{i*}_t$ from (\ref{eq1-3}), (\ref{eq1-4}), (\ref{eq1-5})}
    \STATE {$t \gets t - 1$}
\ENDWHILE 
\STATE{\textbf{Output:} $K^{i*}_t, L^{i*}_t, P^{i*}_t, S^{i*}_t, w^{i*}_t, i \in \mathbb{N}, t \in \mathbb{T}$}
\end{algorithmic}
\end{algorithm}
\begin{theorem}
Consider an unknown-dynamics T-stage game (\ref{isosystem}, \ref{costfun}) with offline data $\mathcal{W}_d$ and initial data $u_{\text{ini}}, y_{\text{ini}}$ generated from an underlying dynamics (\ref{isosystem}).  Suppose that Assumption \ref{A1} holds. If $|\tilde{H}_t(P_{t+1})| \neq 0$ for all 
$t\in \mathbb T$, then the equations (\ref{eq1-1} - \ref{eq1-5}) have a unique solution $\{P^{i*}_t, S^{i*}_t, w^{i*}_t, K^{i*}_t, L^{i*}_t \mid t\in \mathbb T, i\in \mathbb N\}$. In this case, Algorithm \ref{alg:Algorithm1} will output this solution, and $\{u^{i*}_t=K^{i*}_t  U_{t-1} + L^{i*}_t \mid t\in \mathbb T, i\in \mathbb N\}$ is an FNE of the game.
\label{theorem 3}
\end{theorem}
\begin{proof}
    See Appendix.\ref{Proof of Theorem 3}. 
\end{proof}

\section{Infinite-horizon Unknown-dynamics Game}
\label{Infinite-horizon Unknown-dynamics Game}

This section considers the infinite-horizon game with unknown dynamics (\ref{isosystem}) and the cost function defined below. Suppose that $l^i_t = 0$ for all $i \in \mathbb N$ and $t \in \mathbb T$ throughout this section. We first examine the finite-horizon case to support the infinite-horizon game analysis and approximate its FNE.
 In the finite-horizon case, it can be verified that $L^i_t$, $S^i_t$, and $w^i_t$ are all zero, so that the FNE is a linear feedback strategy of the form $u^{i*}_t = K^{i*}_t U_{t-1}$ for all $i \in \mathbb N$ and $t \in \mathbb T$ \cite{1998basarNoncooperativeGame}. Hence, we only need to focus on equations (\ref{eq1-1}) and (\ref{eq1-3}) when the invertibility condition holds. For the infinite-horizon case, this paper restricts attention to FNEs of the linear form \cite{siam2024,TACgame2024} and provides a definition of the FNE similar to Definition 2.1 in \cite{siam2024} based on the term 
\[
U_0(t) = \mathrm{col}(u_{t-1}, \dots, u_{t-T_{\mathrm{ini}}},\; y_{t-1}, \dots, y_{t-T_{\mathrm{ini}}}),
\]
which collects the input-output data from the past $T_{\mathrm{ini}}$ steps at stage $t$ for $t=1,2,...$  When $l^i_t$ is a non-zero constant, we also present the corresponding convergence results in the numerical example in Section~\ref{Numerical Study and Further Analysis} without a theoretical analysis. 
 \begin{align}
J^i(x_1,u)=\frac{1}{2} \sum\limits_{t=1}^{+\infty} [(y_t)^{\top} Q^i y_t + \sum\limits_{j \in N } {(u^j_t)^{\top} R^{ij} u^j_t}](\delta_i)^{t-1}, 
\label{costfun2}
\end{align}

\begin{definition}[FNE in infinite-horizon unknown-dynamics games]
Suppose that the initial data satisfies $T_{\text{ini}} \geq \mathbf{l}(A,C)$ for the infinite-horizon unknown-dynamics game (\ref{isosystem}, \ref{costfun2}).   The strategy profile 
$\{\gamma^{i*}_t=K^{i*}U_0(t)\mid i=1,...,N,~ t=1,2,...\}$ 
is an FNE if and only if for any initial data $U_0(1)=\text{col}(u_{\text{ini}},y_{\text{ini}}) \in \mathcal{B}_{T_{\text{ini}}}$,  $ \forall i \in \mathbb N$, the following inequality  
 \begin{align*}
&J^i(u_{\text{ini}},y_{\text{ini}},K^{i*},K^{(-i)*}) 
\leq J^i(u_{\text{ini}},y_{\text{ini}},K^{i},K^{(-i)*}), 
 \end{align*}
holds for any $K^i \in \mathbb R^{m_i\times T_{ini}}$, where $K^{(-i)*}=\{K^{1*},...,K^{(i-1)*},K^{(i+1)*},...,K^{N*}\}$, $J^i(u_{\text{ini}},y_{\text{ini}},K^{1},...,K^{N})$ denotes the cost of player $i$ when all players adopt the strategy profile 
$\{\gamma^{i*}_t=K^{i}U_0(t)\mid i=1,...,N,~ t=1,2,...\}$. 
\end{definition}

Note that $J^i$ is well-defined, since the initial state is determined by the initial data from Lemma \ref{initialcondition2008}, and the system dynamics are fixed, even though they are unknown to the players. Consequently, for any given input sequence, the corresponding cost is uniquely determined. Next, we define a finite-horizon strategy, ``watching $T$ steps into the future and moving one step now,'' by  
\[
u^{i}_t = K^{i*}_1(T_i) U_0(t), \quad t = 1, 2, \dots,
\]
where each player $i \in \mathbb{N}$ adopts an individual prediction horizon $T_i$ given exogenously. Here, $K^{i*}_1(T_i)$ denotes the first-stage feedback matrix of the unique FNE of the $T_i$-stage unknown-dynamics game under the invertibility condition in Theorem \ref{theorem 3}. 
Under this strategy, each player $i \in \mathbb{N}$ approximates the infinite-horizon game by a $T_i$-stage game at every stage $t$, computes the FNE of the $T_i$-stage unknown-dynamics game, and uses only the first-stage feedback matrix to apply one action. This step-by-step finite-horizon strategy is particularly suitable for infinite-horizon LQ games with unknown dynamics, where limited offline data prevents the direct computation of the exact FNE. It is worth noting that the matrix $K^{i*}_1(T_i)$ is determined solely by the parameters $Q^i$, $\delta_i$, $R^{ij}$ and the offline data, and is independent of the initial data, state, inputs, or outputs. This property follows from Theorem \ref{theorem 3} and is consistent with a similar property in known-dynamics games  \cite{2025finitestrategy}. As a result, $K^{i*}_1(T_i)$ is a constant matrix that only needs to be computed once. Each player $i$ can then implement the finite-horizon strategy by playing $u^i_t = K^{i*}_1(T_i) U_0(t)$ based on the current historical data $U_0(t)$.

This section investigates the following questions:  
(i) whether the feedback matrix of the FNE of the finite-horizon unknown-dynamics game converges as the horizon length tends to infinity;  
(ii) if all players adopt the finite-horizon strategy with individual prediction horizons $T_i$ for each $i \in \mathbb{N}$ in the infinite-horizon unknown-dynamics game, whether the total cost of each player under this strategy converges to the cost under the FNE of the corresponding infinite-horizon known-dynamics game. If so, what is the explicit rate of convergence? 
In the infinite-horizon known-dynamics setting, Huang et al. \cite{2025finitestrategy} shows that the finite-horizon strategy converges in terms of player cost. We now present the following theorem, which establishes that similar convergence results for both the partial feedback matrices and the total cost also hold in the unknown-dynamics setting under suitable conditions. In this theorem, let $K^{i*}_t(T)$  and $\overline{K}^{i*}_t(T)$ denote the feedback matrices corresponding to the unique FNEs of the $T$-stage game with unknown dynamics and known dynamics under the invertibility condition, respectively. In the infinite-horizon game with known dynamics, let $\overline{K}^{i*}$ denote the FNE feedback matrices associated with the limiting matrices of the finite-horizon known-dynamics game \cite{2025finitestrategy}. Hereafter, the FNE of the infinite-horizon known-dynamics game refers to this particular equilibrium.

\begin{theorem}
For the known-dynamics game (\ref{isosystem}, \ref{costfun2}), suppose that system (\ref{isosystem}) is controllable, Assumption 1 in \cite{2025finitestrategy} holds. Under the FNE associated with the 
limiting matrices in  the infinite-horizon known-dynamics game, the total cost of player $i$ is denoted by $J^i(x_1) \triangleq J^i$, for all $i \in \mathbb{N}$. 
For the unknown-dynamics game (\ref{isosystem}, \ref{costfun2}), suppose that sufficient offline input/output data are available such that Assumption \ref{A1} holds for any $L \in \mathbb{N}_{+}$, and $|\tilde{H}_t(P_{t+1})| \neq 0$ for $t = 1, \dots, T$, with $P_{T+1} = 0$ for any $T \in \mathbb{N}_+$. 

In the infinite-horizon game with unknown dynamics, suppose all players adopt the finite-horizon strategies of the form $u^{i}_t = K^{i*}_1(T_i) U_0(t)$, where each player $i \in \mathbb{N}$ uses an individual prediction horizon $T_i$ given exogenously, $K^{i*}_1(T_i)$ denotes the first-stage feedback matrix of the unique FNE of the $T_i$-stage unknown-dynamics game. 
The total cost of player $i$ under the finite-horizon strategy is denoted by $\tilde{J}^i(u_{\text{ini}}, y_{\text{ini}})(T_1, \dots, T_n) $ $\triangleq \tilde{J}^i(u_{\text{ini}}, y_{\text{ini}})$. 
Then, we have
\begin{enumerate}
    \item For any $t \in \mathbb{N}_+$, $\lim\limits_{T \rightarrow +\infty} (K^{i*}_t(T))_{:,u_k}=\overline K^{i*} A^{t-k-1}B$ for $k = 1, \dots, t - 1$.
    
    \item $\lim\limits_{T_m \rightarrow +\infty} \tilde{J}^i(u_{\text{ini}}, y_{\text{ini}}) = J^i$, where $T_m = \min(T_1, \dots, T_N)$.
    
    \item If $\lim\limits_{T \rightarrow +\infty} K^{i*}_1(T) \triangleq  K^{i*}$ exists and is finite for all $i \in \mathbb{N}$, then 
    $$\{ u^{i*}_t=K^{i*}U_0(t)\mid i\in \mathbb N, ~t=1,2,...\}$$ 
    is an FNE of the infinite-horizon unknown-dynamics game, and we have 
    \[
    \tilde{J}^i(u_{\text{ini}}, y_{\text{ini}}) - J^i = O\left( \max\limits_{i \in \mathbb{N}} \|K^{i*}_1(T_i) - K^{i*} \|_2 \right), \quad \text{as } T_m \rightarrow +\infty.
    \]
\end{enumerate}
\label{theorem 4}
\end{theorem}

\begin{proof}
    See Appendix.\ref{Proof of Theorem 4}. 
\end{proof}

\begin{algorithm}
\caption{Approximate Algorithm for \(+\infty\)-stage Unknown-dynamics Game}\label{alg:Algorithm2}
\begin{algorithmic}
\STATE{\textbf{Input:} Offline data \(\mathcal{W}_d\), initial data \( u_{\text{ini}}, y_{\text{ini}} \), \( T_{\text{limit}} \), 
\( Q^i, \delta_i, R^{ij} \), \( T = 1 \), \( \epsilon = 0.01 \), \( k = 0 \), \( M = 1000 \), \( K^{i*}_1 = \mathbf{0} \), \( \tilde{J}^{i*} = 0 \),  \( U(0) = (u_{\text{ini}}, y_{\text{ini}}) \), \( i,j \in \mathbb{N} \)}

\STATE{ Compute \( K^{i*}_1(T) ~\text{and}~ K^{i*}_1(T+1)\) using Algorithm \ref{alg:Algorithm1}, \( K^{i*}_1 \gets K^{i*}_1(T), ~ i \in \mathbb{N} \)}

\WHILE{$\max\limits_{i \in \mathbb{N}} \| K^{i*}_1(T) - K^{i*}_1(T+1) \|_2 > \epsilon$ and \( T \leq T_{\text{limit}} \)}
    \STATE{ \( T \gets T + 1 \)}
    \STATE{Compute \( K^{i*}_1(T) ~\text{and}~ K^{i*}_1(T+1), ~ i \in \mathbb{N} \) using Algorithm \ref{alg:Algorithm1}}
\ENDWHILE

\STATE{ \( K^{i*}_1 \gets K^{i*}_1(T), ~ i \in \mathbb{N} \)}

\WHILE{$k \leq M$}
    \STATE{ \( \tilde{J}^{i*} \gets \tilde{J}^{i*} + U(k)^{\top} \left[ (G_1 \mathbf{K}_1^*)^{\top} Q^i (G_1 \mathbf{K}_1^*) + \sum_{j \in \mathbb{N}} (K^{j*}_1)^{\top} R^{ij} (K^{j*}_1) \right] U(k) (\delta_i)^{k-1} \)}
    \STATE{ \( k \gets k + 1 \), \( u^i_k = K^{i*}_1 U(k-1) \), \( y_k = G_1 \mathbf{K}_1^* U(k-1) \)}
    \STATE{\( U(k) \gets (u_{k-T_{\text{ini}}+1}, \dots, u_{k-1}, u_k, y_{k-T_{\text{ini}}+1}, \dots, y_{k-1}, y_k) \)}
\ENDWHILE

\STATE{\textbf{Output:} \( K^{i*}_1, \tilde{J}^{i*}, i \in \mathbb{N} \)}
\end{algorithmic}
\end{algorithm}

Theorem \ref{theorem 4} demonstrates the convergence property of the finite-horizon strategies. Although the columns of the feedback matrices $K^{i*}_t(T)$ corresponding to $u_k$ for $k = 1, \dots, t - 1$ converge as $T \rightarrow +\infty$, it is unclear whether $(K^{i*}_t(T))_{:,(u_{\text{ini}}, y_{\text{ini}})}$ also converges. The second part of the theorem shows that each player's total cost under the finite-horizon strategy in the infinite-horizon game with unknown dynamics converges to the total cost under the particular FNE of the corresponding game with known dynamics. This result addresses two major challenges in approximating the FNE: the infinite horizon and the unknown dynamics of the game. Furthermore, we prove that the convergence rate of the cost difference is no slower than that of the distance between the corresponding feedback matrices.  

However, in certain scenarios, it is unclear whether the underlying dynamics satisfy the assumptions in Theorem \ref{theorem 4}, and the available offline data may be limited. Thus, we provide Algorithm \ref{alg:Algorithm2} to help players obtain approximate strategies. If the relevant assumptions hold, this algorithm yields finite-horizon strategies and their associated costs, where $T_{\text{limit}}$ denotes the maximum prediction horizon that offline data can support, i.e., the largest value of $T$ for which Assumption \ref{A1} holds. The algorithm may fail in some cases, such as when the offline data is insufficient or the system dynamics do not satisfy Assumption 1 in \cite{2025finitestrategy}.

\section{Numerical Study and Further Analysis}
\label{Numerical Study and Further Analysis}

In this section, we present simulations for two-player games based on the following linear system and objective functions. The goal is to illustrate the convergence of the feedback matrices $K^{i*}_1(T), L^{i*}_1(T)$ and the total cost under the finite-horizon strategy as $T \rightarrow +\infty$. 
\begin{equation}
\label{isosystem4}
\begin{aligned}
x_{t+1} &= \begin{bmatrix}
0.9 & 0.2   & -0.6\\
-0.4   & 0.8 & 0.2\\
0.5   & 0.3   & 0.1
\end{bmatrix}x_t + \begin{bmatrix}
1 & -0.3\\
-2 & 0.5\\
-0.3 & 0.3
\end{bmatrix}\begin{bmatrix}
u^1_t \\
u^2_t 
\end{bmatrix} ,\\
y_t &= \begin{bmatrix}
-1 & 0.3 & -0.2\\
-0.1 & 0.5  & 1
\end{bmatrix}x_t + \begin{bmatrix}
0.1  & 0.5\\
-0.4  & 0.1
\end{bmatrix}\begin{bmatrix}
u^1_t \\
u^2_t 
\end{bmatrix}.
\end{aligned}
\end{equation}

\begin{equation}
\label{costfun4}
\begin{aligned}
&J_1(x_1,u)=\frac{1}{2} \sum_{t=1}^{T} [(y_t-\scalebox{0.9}{$\begin{bmatrix} -1\\ 0.3 \end{bmatrix}$})^{\top} \begin{bmatrix} 2 & 0.2\\ 0.2 & 2 \end{bmatrix} (y_t-\scalebox{0.9}{$\begin{bmatrix} -1\\ 0.3 \end{bmatrix}$}) + 0.5 (u_t^1)^2 - 0.1 (u_t^2)^2] (0.8)^{t-1}
\\
&J_2(x_1,u)=\frac{1}{2} \sum_{t=1}^{T} [(y_t-\scalebox{0.9}{$\begin{bmatrix} 0.4\\ -0.3 \end{bmatrix}$})^{\top} \begin{bmatrix} 3 & 0.5\\ 0.5 & 3 \end{bmatrix} (y_t-\scalebox{0.9}{$\begin{bmatrix} 0.4\\ -0.3 \end{bmatrix}$}) - 0.3 (u_t^1)^2 + (u_t^2)^2] (0.9)^{t-1}
\end{aligned}
\end{equation}

Consider the games with dynamics (\ref{isosystem4}) and objective function (\ref{costfun4}), where $u_t=\text{col}(u^1_t, u^2_t) \in \mathbb{R}^2$, $x_t \in \mathbb{R}^3$, and $y_t \in \mathbb{R}^2$. Let $x_1 = \text{col}(1.885, -3.208, -0.922)$ be randomly generated. For simplicity, let $T_{\text{ini}} = \mathbf{l}(A,C) = 2$. The initial data are $u_{\text{ini}} = \text{col}(u_{-T_{\text{ini}}+1}, ..., u_0) = (u^1_{-1}, u^2_{-1}, u^1_0, u^2_0) = \text{col}(-0.640, -4.741, 0.497, -0.647)$, and $y_{\text{ini}} = \text{col}(y_{-T_{\text{ini}}+1}, ..., y_0) = (y_{-1}, y_0) = \text{col}(-1.534, -5.884, -0.637, -3.849)$. It can be verified that the pair $(u_{\text{ini}}, y_{\text{ini}})$ is matched with $x_1$. We consider a 500-stage offline input/output data $u^d_t, y^d_t$ for $t = 1, ..., 500$, where $u^d_t$ is randomly generated from a uniform distribution on $[-5, 5]$, $x^d_1 = \mathbf{0}$, and $y^d_t$ is generated using the dynamics (\ref{isosystem4}) with $x^d_1$ and $u^d_t$. Here, the superscript $d$ indicates offline data. When $T \leq 100$ and $L = T + T_{\text{ini}}$, Assumption \ref{A1} is satisfied by these data.

\begin{figure}
    \centering
    \includegraphics[width=1\linewidth]{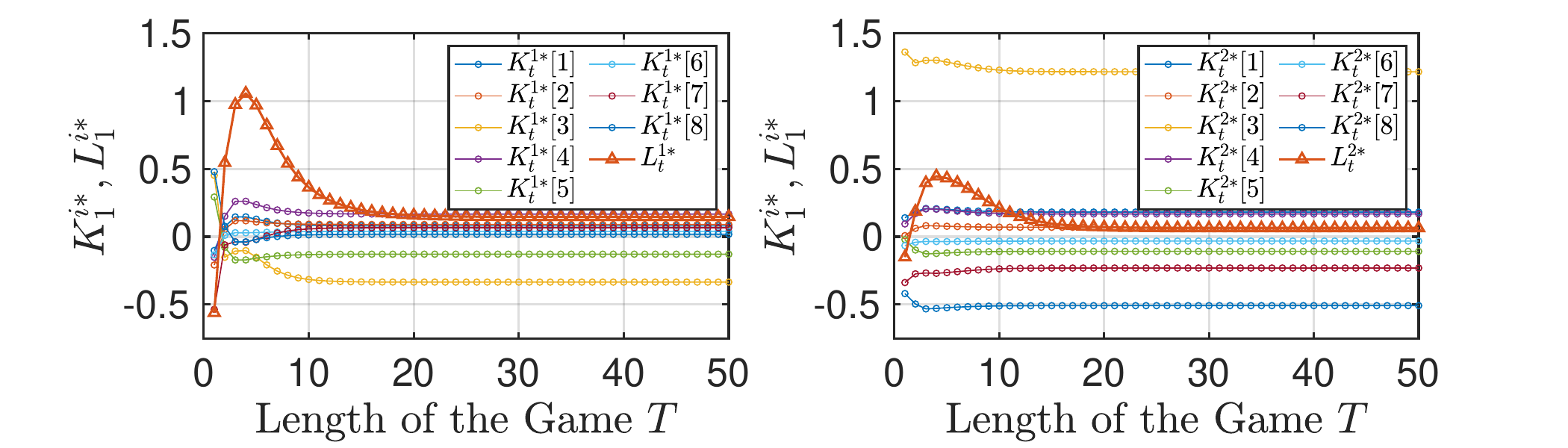}
  \caption{The unique FNE matrices at stage 1, $K^{i*}_1(T) \triangleq K^{i*}_1$ and $L^{i*}_1(T) \triangleq L^{i*}_1$, of the unknown-dynamics game with (\ref{isosystem4}), (\ref{costfun4}), and different game lengths $T = 1, \dots, 50$, for players $i = 1, 2$ and $t = 1, \dots, 100$. Here, $K^{i*}_1(T)[j] \triangleq K^{i*}_1[j]$, $j = 1, \dots, 8$, denotes all the components of $K^{i*}_1(T)$.}

    \label{fig.KL_unknown_T}
\end{figure}
To examine the convergence of the FNE matrices $K^{i*}_1(T) \in \mathbb{R}^{8\times 1}$ and $L^{i*}_1(T) \in \mathbb{R}$ for the unknown-dynamics game, we compute the FNEs of the $T$-stage games for $T = 1, \dots, 50$ using Algorithm~\ref{alg:Algorithm1}. Fig.~\ref{fig.KL_unknown_T} shows the 8 components of $K^{i*}_1(T)$ and $L^{i*}_1(T)$, which exhibit convergence as the game length $T$ tends to infinity, even when the objective trajectories are non-zero. The FNE of the infinite-horizon unknown-dynamics game, which corresponds to the finite-horizon strategy approximation, is also shown below. Here, $U_0(t) = \text{col}(u_{t-1}, \dots, u_{t - T_{\text{ini}}}, y_{t-1}, \dots, y_{t - T_{\text{ini}}})$ represents the data from the past $T_{\text{ini}}$ steps at stage $t=1,2,...$
\begin{align*}
u^{1*}_t&= K^{1*}U_0(t)+L^{1*}\\
&= [0.079, 0.090, -0.335, 0.167, -0.129, 0.039, 0.067, 0.018]
U_0(t) + 0.146,\\
u^{2*}_t&= K^{2*}U_0(t)+L^{2*}\\
&= [0.182, 0.069, 1.217, 0.168, -0.108, -0.032, -0.231, -0.509]
U_0(t)  + 0.064.
\end{align*}

\begin{figure}
    \centering
    \includegraphics[width=0.7\linewidth]{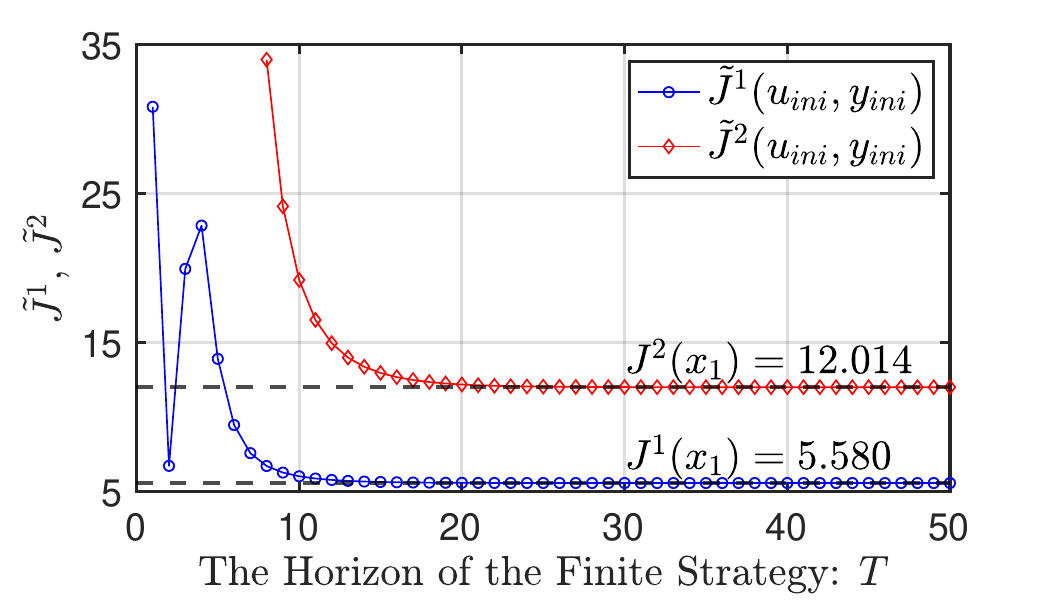}
    \caption{The total cost $\tilde{J}^i(u_{\text{ini}}, y_{\text{ini}})(T) \triangleq \tilde{J}^i(u_{\text{ini}}, y_{\text{ini}})$ under the finite-horizon  strategy in the infinite-horizon unknown dynamics game with (\ref{isosystem4}) and (\ref{costfun4}), for players $i=1,2$ and $T=1,...,50$. The horizontal black dashed lines represent the total cost $J^i(x_1)$ under the FNE of the infinite-horizon known-dynamics game for players $i=1,2$, where $x_1$ is matched with $u_{\text{ini}}, y_{\text{ini}}$.}
    \label{fig.costTunknown}
\end{figure}

For each prediction horizon $T = 1, 2, \dots, 50$ in the finite-horizon strategy, we compute the corresponding cost of player $i$ in the infinite-horizon game under these strategies, using Algorithm~\ref{alg:Algorithm2}. This cost is denoted by $\tilde{J}^i(u_{\text{ini}}, y_{\text{ini}})(T) \triangleq \tilde{J}^i(u_{\text{ini}}, y_{\text{ini}})$.  The results are shown in Fig.~\ref{fig.costTunknown}, where $\tilde{J}^i(u_{\text{ini}}, y_{\text{ini}})$ converges to $J^i(x_1)$ as $T$ tends to infinity. Here, $J^i(x_1)$ denotes the total cost of player $i$ under the unique FNE of the infinite-horizon game with known dynamics, where the initial state $x_1$ is matched  with the initial data $u_{\text{ini}}, y_{\text{ini}}$.
When the differences $\|K^{i*}_1(T) - K^{i*}_1(T+1)\|_2$ for $i = 1,2$ become small, e.g., at $T = 20$ in Fig. \ref{fig.KL_unknown_T}, the corresponding cost under the strategy ``watching 20 steps into the future and moving one step now'' is already close to $J^i(x_1)$. Even without knowing the system dynamics and given  100-length offline data and 2-length initial data, players can approximate the infinite-horizon FNE of the unknown-dynamics game using finite-horizon strategies.

\section{Conclusion}
\label{conclusion}
We use a data-driven and behavioral system approach to analyze the discrete-time linear-quadratic game with unknown input/output/state dynamics and offline input/output data. We present the detailed relationships between the FNEs of the unknown- and known-dynamics games, including the equivalence of their FNE sets. Additionally, we propose the coupled equations based on offline data, which provide the sufficient and necessary conditions for the existence and uniqueness of FNEs in the finite-horizon unknown-dynamics game. We also provide a sufficient condition along with an  algorithm to simplify the computation of the FNE. Finally, we prove that, in the infinite-horizon unknown-dynamics game, the finite-horizon strategy of ``watching $T$ steps into the future and moving one step now,'' supported by the limited offline data, exhibits convergence properties in terms of partial feedback matrices of the FNE of the finite-horizon game and related total costs. Furthermore, the convergence rate of the cost difference is no slower than that of the distance between the corresponding feedback matrices. Some remaining problems include the unknown-dynamics game with unknown objective functions of other players, noisy offline data, non-linear dynamics, and how these factors can naturally characterize more realistic competition scenarios.

\appendix
\section{Proof of \cref{theorem 1}}
\label{Proof of Theorem 1}

\begin{proof}
{\bf (a) Affine Form Strategy: } We first prove that any FNE of the $T$-stage unknown-dynamics game (\ref{isosystem}, \ref{costfun}) has an affine form, and the cost of each player under the FNE has a quadratic form.

\begin{equation}
    \label{program T}
\begin{aligned}    \min\limits_{u^i_T} &\frac{1}{2} [(y_T-l^i_T)^{\top} Q^i (y_T-l^i_T) + \sum\limits_{j \in N } {(u^j_T)^{\top} R^{ij} u^j_T}](\delta_i)^{T-1} \\  
&s.t. \begin{bmatrix}
 U_p \\
 Y_p \\
 U_f \\
 Y_f
 \end{bmatrix}g =
 \begin{bmatrix}
 u_{\text{ini}} \\
 y_{\text{ini}} \\
 u \\
 y
 \end{bmatrix}, ~~\left\{\begin{aligned}
    &u=\text{col}(u_1,...,u_T)\\ &y=\text{col}(y_1,...,y_T)\\&u_T=\text{col}(u_T^1,...,u_T^N)
 \end{aligned}\right.
    \end{aligned}
    \end{equation}

At stage $T$, player $i$ considers the optimal program (\ref{program T}) when $(u_1,...,u_{T-1}),$ $ (y_1,...,y_{T-1})$ and $u^j_T, j\in \mathbb N\setminus {i}$ are given. Thus, this program only has three variables: $g, u^i_T, y_T$. Since $x_1$ is uniquely determined by $u_{\text{ini}}$ and $y_{\text{ini}}$ from Lemma \ref{initialcondition2008}, all outputs $y_1,\dots,y_T$ are uniquely determined by the inputs $u_1,\dots,u_T$. Thus, we only need to focus on the equations without $y$:
$$\begin{bmatrix}
 U_p \\
 Y_p \\
 U_f 
 \end{bmatrix}g =
 \begin{bmatrix}
 u_{\text{ini}} \\
 y_{\text{ini}} \\
 u 
 \end{bmatrix}.$$
  Thus, $g = M_T^{\dag} \operatorname{col}(u_{\text{ini}}, y_{\text{ini}}, u)$, where $M_T = \operatorname{col}(U_p, Y_p, U_f)$ and $M_T^{\dag}$ is the pseudoinverse of $M_T$.
We recall the following notations:
$Y_f = \operatorname{col}(Y_{f1}, Y_{f2}, \dots, Y_{fT})$,
$Y_{fT} = (Y_f)_{y_T,:}$,
$y_T = Y_{fT} g$,
$M = Y_{fT} M_T^{\dag}$,
$M = [M_{:,u_{\text{ini}}}, M_{:,y_{\text{ini}}}, M_{:,u_1}, \dots, M_{:,u_N}]$.
Then, program (\ref{program T}) is equivalent to
    \begin{equation}
    \begin{aligned}&\min\limits_{u^i_T}
    \frac{1}{2} [\mathbf{M}^{\top} Q^i\mathbf{M} 
     + \sum\limits_{j \in N } {(u^j_T)^{\top} R^{ij} u^j_T}](\delta_i)^{T-1}\\     &\mathbf{M}=M_{:,u_{\text{ini}}}u_{\text{ini}}+M_{:,y_{\text{ini}}}y_{\text{ini}}+M_{:,u_{1}}u_{1}+...+M_{:,u_{T}}u_{T}-l^i_T.
     \end{aligned}
    \label{T-cost}
    \end{equation}
    
 Since $Q^i$ is semidefinite and $R^{ii}$ is positive definite, (\ref{T-cost}) is strictly convex in $u^i_T$. Therefore, player $i$'s best response exists and is unique, and it equals the unique root of the derivative of (\ref{T-cost}).
Using $\frac{\partial}{\partial x} \left( (Ax + b)^{\top} Q (Ax + b) \right) = 2A^{\top} Q (Ax + b)$, the derivative of (\ref{T-cost}) with respect to $u^i_T$ yields a linear equation in $u_{\text{ini}}, y_{\text{ini}}, u$
    \begin{align*}
    &(M_{:,u_{T}^i})^{\top} Q^i \mathbf{M} 
     +  R^{ii} u^i_T=0,
     \end{align*}
and the coefficients of $u_{\text{ini}},y_{\text{ini}}, u_1, \ldots, u_T$ come solely from the offline data and the parameters in the players' objectives. By combining the $N$ equations from all players, we obtain
\begin{align*}
F_T \text{col}(u^1_T,...,u^N_T) = G_T \text{col}(u_{\text{ini}},y_{\text{ini}}, u_1,..., u_{T-1}) + H_T=G_T U_{T-1}+ H_T,
\end{align*}
    where $F_T \in \mathbb{R}^{m \times m}$. These linear equations imply that the FNE at stage $T$ is an affine function of $U_{T-1}$.  
When all players adopt one of the solutions   $u^{i*}_T(U_{T-1})=K^{i*}_TU_{T-1}+L^{i*}_T$, the cost of player $i$ at stage $T$ is as follows 
    \begin{align*}
    V^i_T(U_{T-1})
    =&\frac{1}{2} [(M \text{col}(U_{T-1}, u^*_T(U_{T-1}))-l^i_T)^{\top} Q^i (M \text{col}(U_{T-1}, u^*_T(U_{T-1}))-l^i_T) \\
    + &\sum\limits_{j \in N } {(u^{j*}_T(U_{T-1}))^{\top} R^{ij} u^{j*}_T(U_{T-1})}](\delta_i)^{T-1}
    \\ \triangleq &\frac{1}{2} (U_{T-1})^{\top} P_T^{i*} U_{T-1} + (S^{i*}_T)^{\top} (U_{T-1}) + w^{i*}_T.
    \end{align*}

 At stage $T-1$, suppose that at least one FNE exists at stage $T$, given by $u^{i*}_T(U_{T-1})=K^{i*}_TU_{T-1}+L^{i*}_T$, and all players adopt this strategy. 
Player $i$ then considers the  program below when $(u_1, \ldots, u_{T-2})$, $(y_1, \ldots, y_{T-2})$, and $u^j_{T-1}$ for all $j \in \mathbb{N} \setminus {i}$ are given. Since $V^i_T(U_{T-1})$ takes a quadratic form, we similarly obtain that the FNE at stage $T-1$ has an affine form. The $(T-1)$-stage cost is a quadratic function with respect to $U_{T-2}$.
By repeating the same analysis for stages $T-2, T-3, \ldots, 1$, the proof of part (a) is complete. 
   \begin{equation}
    \label{program T-1}
\begin{aligned}   \min\limits_{u^i_{T-1}} &\frac{1}{2} (\delta_i)^{T-2} (y_{T-1}-l^i_{T-1})^{\top} Q^i (y_{T-1}-l^i_{T-1}) 
     + V^i_T(U_{T-1})\\
     +&\frac{1}{2} (\delta_i)^{T-2}\sum\limits_{j \in N } {(u^j_{T-1})^{\top} R^{ij} u^j_{T-1}}\\  
    &s.t.\begin{bmatrix}
 U_p \\
 Y_p \\
 (U_f)_{1:T-1,:} \\
 (Y_f)_{1:T-1,:}
 \end{bmatrix}g =
 \begin{bmatrix}
 u_{\text{ini}} \\
 y_{\text{ini}} \\
 u \\
 y
 \end{bmatrix}, ~~\left\{\begin{aligned}
&u=\text{col}(u_1,...,u_{T-1}),\\
&y=\text{col}(y_1,...,y_{T-1}),\\
&u_{T-1}=\text{col}(u_{T-1}^1,...,u_{T-1}^N).
\end{aligned}\right.   
\end{aligned}
\end{equation}

{\bf (b) Multiple Feasible Set Equivalence:} Now we prove that the unknown-dynamics and known-dynamics games share the same set of FNEs in terms of the input variables. Since the cost functions for all players are the same in both settings, it suffices to prove the equivalence of their feasible sets. 
In the unknown-dynamics game, each player $i \in \mathbb{N}$ aims to minimize her cost function (\ref{costfun}) subject to the feasible set represented by the behavioral system 
\begin{equation*}
\begin{aligned}    
\begin{bmatrix}
 U_p \\
 Y_p \\
 U_f \\
 Y_f
 \end{bmatrix}g =
 \begin{bmatrix}
 u_{\text{ini}} \\
 y_{\text{ini}} \\
 u \\
 y
 \end{bmatrix}, ~~\left\{\begin{aligned}
    &u=\text{col}(u_1,...,u_T)\\ &y=\text{col}(y_1,...,y_T)\\&u_t=\text{col}(u_t^1,...,u_t^N),t\in \mathbb T. 
 \end{aligned}\right.
    \end{aligned}
    \end{equation*} 
In the known-dynamics game, each player $i \in \mathbb{N}$ aims to the same cost function subject to the feasible set represented by the linear  system 
\begin{equation}
    \label{program T known}
\begin{aligned}\left\{\begin{aligned}
 &x_{t+1} = Ax_t+ Bu_t, x_1 ~\text{is given} \\
 &y_t = Cx_t+ Du_t, t=-T_{\text{ini}}+1,...,T\\
 &\text{col}(u_k, {-T_{\text{ini}}+1 \leq k \leq 0})=u_{\text{ini}} \\
 &\text{col}(y_k, {-T_{\text{ini}}+1 \leq k \leq 0})=y_{\text{ini}}
 \end{aligned}   \right. \end{aligned}
    \end{equation}

By Lemma \ref{initialcondition2008} and Assumption \ref{A1}, $x_1$ is uniquely determined by $u_{\text{ini}}$ and $y_{\text{ini}}$. Therefore, in the feasible set represented by the linear system, $x_1$ can be omitted.
According to Lemma \ref{nonpersistentexcitation2022}, we have 
$\text{Image}(\mathcal{H}_{T_{\text{ini}}+T}(w^d))=\mathcal{B}_{T_{\text{ini}}+T}$. Hence, the feasible set represented by the behavioral system is  $\{(u, y) \mid\text{col}(u_{\text{ini}}, u, y_{\text{ini}}, y) \in \mathcal{B}_{T_{\text{ini}}+T}\}$.
By the definition of the i/o/s  representation, we obtain 

\[
\begin{aligned}
&\{(u, y) | \text{col}(u_{\text{ini}}, u, y_{\text{ini}}, y) \in \mathcal{B}_{T_{\text{ini}}+T}\}\\
= &\{(u, y)\mid 
\begin{aligned}
&\begin{bmatrix}
 x_{t+1} = Ax_t+ Bu_t \\
 y_t = Cx_t+ Du_t\\
 \text{col}(u_k, {-T_{\text{ini}}+1 \leq k \leq 0})=u_{\text{ini}} \\
 \text{col}(y_k, {-T_{\text{ini}}+1 \leq k \leq 0})=y_{\text{ini}} 
 \end{bmatrix}, -T_{\text{ini}}+1\leq t \leq T\},
\end{aligned} 
\end{aligned}
\]
which indicates the equivalence of the two feasible sets. 
This completes the proof.
\end{proof}

\section{Proof of Theorem \ref{theorem 2}}
\label{Proof of Theorem 2}

\begin{proof}
 {\bf (a) Part (I): An FNE $\Rightarrow$ A Solution. } First, we prove that for every FNE $\{u^{i*}_t=K^{i*}_t  U_{t-1} + L^{i*}_t\mid t\in \mathbb T, i\in \mathbb N\}$, there exists a set $\{P^{i*}_t, S^{i*}_t, w^{i*}_t| t\in \mathbb T, i\in \mathbb N\}$  such that $\{K^{i*}_t, L^{i*}_t,P^{i*}_t, S^{i*}_t, w^{i*}_t| t\in \mathbb T, i\in \mathbb N\}$  is a solution to equations (\ref{eq1-1a}, \ref{eq1-1b}, \ref{eq1-2} - \ref{eq1-5}).
We begin by setting the terminal conditions: $P_{T+1}^{i*}=0, ~S_{T+1}^{i*}=0, ~w_{T+1}^{i*}=0.$  
Since $K^{i*}_T$ and $L^{i*}_T$ are given, we can directly compute $P_T^{i*}, S_T^{i*}, w_T^{i*}$ using equations (\ref{eq1-3} - \ref{eq1-5}). This procedure can be repeated backward to define $\{P_{t}^{i*}, ~S_{t}^{i*}, ~w_{t}^{i*}| t\in \mathbb T\}$. By construction, $\{P^{i*}_t, S^{i*}_t, w^{i*}_t, K^{i*}_t, L^{i*}_t| t\in \mathbb T, i\in \mathbb N\}$ satisfies equations (\ref{eq1-3} - \ref{eq1-5}). Therefore, it remains to verify that it also satisfies equations (\ref{eq1-1a}, \ref{eq1-1b}, \ref{eq1-2}).

At stage $T$, let equations (\ref{eq1-3}), (\ref{eq1-4}), and (\ref{eq1-5}) represent the quadratic, linear, and constant terms with respect to $U_{T-1}$, respectively. Specifically, multiply (\ref{eq1-3}) on the left by $0.5 (\delta_i)^{T-1} U_{T-1}^{\top}$ and on the right by $U_{T-1}$, multiply (\ref{eq1-4}) on the left by $(\delta_i)^{T-1} U_{T-1}^{\top}$, and multiply (\ref{eq1-5}) by $(\delta_i)^{T-1}$. 
Then, summing the three resulting expressions yields:
\begin{equation}
\begin{aligned}
&V^{i*}_T(U_{T-1})\\&\triangleq (\delta_i)^{T-1}\left(\frac{1}{2}U^{\top}_{T-1} P^{i*}_{T} U_{T-1} +   U_{T-1}^{\top}(S^{i*}_T) + w^{i*}_T\right)\\
    &= \frac{1}{2} (\delta_i)^{T-1} \left(G_T(\mathbf{K}_T^{*} U_{T-1} + \mathbf{L}_T^*) - l^i_T\right)^{\top} Q^i \left(G_T(\mathbf{K}_T^{*} U_{T-1} + \mathbf{L}_T^*) - l^i_T\right) \\
    &+\sum\limits_{j=1}^N \frac{1}{2}  (\delta_i)^{T-1} \left(K^{j*}_T  U_{T-1} + L^{j*}_T\right)^{\top} R^{ij} \left(K^{j*}_T  U_{T-1} + L^{j*}_T\right) \\
    &+ \frac{1}{2} (\delta_i)^{T}(\mathbf{K}_T^{*} U_{T-1} + \mathbf{L}_T^*)^{\top} P^{i*}_{T+1} (\mathbf{K}_T^{*} U_{T-1} + \mathbf{L}_T^*) \\
    &+ (\delta_i)^{T}  (\mathbf{K}_T^{*} U_{T-1} + \mathbf{L}_T^*)^{\top} (S^{i*}_{T+1}) + (\delta_i)^{T} w^{i*}_{T+1},
    \end{aligned}
\label{proof eq2-1}
\end{equation}
for any $U_{T-1}\in \mathcal{B}_{T_{\text{ini}}} \times R^{m(T-1)}$, where $\mathbf{K}^*_T = \begin{bmatrix}
I \\
K^*_T
\end{bmatrix},  \mathbf{L}_T^* = \begin{bmatrix}
0 \\
L^*_T
\end{bmatrix},$  $K_T^*=\text{col}(K_T^{1*},...,$ $K_T^{N*}),$ $ L_T= \text{col}(L_T^{1*},...,L_T^{N*}),$ and $I$ is the identity matrix. To simplify this equation, we need the  equations 
\begin{equation*}
    \begin{aligned}
    U^*_{T} = 
    \begin{bmatrix}
    U_{T-1} \\
    u^{1*}_T\\
    \vdots \\
    u^{N*}_T
    \end{bmatrix} = 
    \begin{bmatrix}
    I \\
    K^{1*}_T \\
    \vdots \\
    K^{N*}_T
    \end{bmatrix}U_{T-1} +
    \begin{bmatrix}
    0 \\
    L^{1*}_T \\
    \vdots \\
    L^{N*}_T 
    \end{bmatrix} =
    \begin{bmatrix}
    I \\
    K^{*}_T
    \end{bmatrix}U_{T-1} +
    \begin{bmatrix}
    0 \\
    L^*_T
    \end{bmatrix}=
    \mathbf{K}_T^* U_{T-1} + \mathbf{L}_t^*
    \end{aligned}, 
\end{equation*} 
and 
\begin{equation*}
    \begin{aligned}
    \begin{bmatrix}
    U_p \\
    Y_P \\
    U_f \\
    Y_f
    \end{bmatrix}g = 
    \begin{bmatrix}
    u_{\text{ini}} \\
    y_{\text{ini}} \\
    u\\
    y
    \end{bmatrix} \Rightarrow 
    y_T = Y_{fT}g = Y_{fT}
    \begin{bmatrix}
    U_p \\
    Y_P \\
    U_f 
\end{bmatrix}^{\dag} 
    \begin{bmatrix}
    u_{\text{ini}} \\
    y_{\text{ini}} \\
    u
    \end{bmatrix} =
    Y_{fT}M^{\dag}_T U_T = G_T U_T
    \end{aligned}.
\end{equation*}
The latter equation holds when $u^{i*}_T=K^{i*}_T x_t + L^{i*}_T$ for $i \in \mathbb{N}$, then (\ref{proof eq2-1}) is equal to 
\begin{equation}
    \begin{aligned}
    &V^{i*}_T(U_{T-1})\\
    &=\frac{(\delta_i)^{T-1}}{2}[(y_T(u^*_T)-l^i_T)^{\top} Q^i (y_T(u^*_T)-l^i_T) + \sum\limits_{j \in \mathbb N}(u^{j*}_T)^{\top} R^{ij}u^{j*}_T]+V_{T+1}^{i*}(U_T^*) \\&=\min\limits_{u^i_T} \frac{(\delta_i)^{T-1}}{2}(y_T(u^i_t,u^{(-i)*}_T)-l^i_T)^{\top} Q^i (y_T(u^i_t,u^{(-i)*}_T)-l^i_T) +V_{T+1}^{i*}(U_T^*) \\&+ \frac{(\delta_i)^{T-1}}{2}[(u^{i}_T)^{\top} R^{ii}u^{i}_T + \sum\limits_{j \neq i}(u^{j*}_T)^{\top} R^{ij}u^{j*}_T]  
    \\
    &\triangleq \min\limits_{u^i_T} H^{i*}_T(u^i_T), i\in \mathbb N,  
    \end{aligned}
\label{proof eq2-2}
\end{equation} 
 where $y_T(u_T^*)$ denotes the output when the inputs are $U_{t-1}, u^*_T$, and $y_T(u^i_t, u^{(-i)*}_T)$ is defined similarly.  The second equality in (\ref{proof eq2-2}) follows from the definition of an FNE, namely, that $u^{i*}_T$ is the best response to $u^{j*}_T$ for all $j \neq i$, which also implies that 
\begin{equation}
    \begin{aligned}
    &\frac{\partial H^{i*}_T}{\partial u^i_T} (u^{i*}_T)\\
    &= (\delta_i)^{T-1}[(G_T)_{:,u^i_T}]^{\top} Q^i (G_T (\mathbf{K}_T^{*} U_{T-1} + \mathbf{L}_t^*) - l^i_T) \\&+ (\delta_i)^{T-1} R^{ii} (K_T^{i*} U_{T-1} + L_T^{i*}) + (\delta_i)^{T}(P^{i*}_{T+1})_{u^i_T,:} (\mathbf{K}_T^{*} U_{T-1} + \mathbf{L}_t^*) \\&+ (\delta_i)^{T}(S^{i*}_{T+1})_{u^i_T,:}  = 0
    \end{aligned}
\label{proof eq2-3}
\end{equation} 
for any $U_{T-1} \in \mathcal{B}_{T_{\text{ini}}} \times \mathbb{R}^{m(T-1)}$, and we have 
$$\mathbf{K}_T^{*} U_{T-1} + \mathbf{L}_t^* = (\mathbf{K}_T^{*})_{:,(u_{\text{ini}},y_{\text{ini}})}  \text{col}(u_{\text{ini}},y_{\text{ini}}) + \sum\limits_{j=1}^{T-1} (\mathbf{K}_T^{*})_{:,u_j} u_j + \mathbf{L}_t^*.$$ 
By comparing the coefficients of $\mathrm{col}(u_{\text{ini}}, y_{\text{ini}})$, each $u_j$ with $j \leq T-1$, and the constant terms on both sides of (\ref{proof eq2-3}), we obtain equations (\ref{eq1-1a}), (\ref{eq1-1b}), and (\ref{eq1-2}) at stage $T$. At stage $T-1$, similarly, multiply equation (\ref{eq1-3}) on the left by $0.5(\delta_i)^{T-2} U_{T-2}^{\top}$ and on the right by $U_{T-2}$, multiply equation (\ref{eq1-4}) on the left by $(\delta_i)^{T-2} U_{T-2}^{\top}$, and multiply equation (\ref{eq1-5}) by $(\delta_i)^{T-2}$. Then, summing these three terms yields 
    \begin{equation}
    \begin{aligned}
    \label{proof eq2-7}
    &V^{i*}_{T-1}(U_{T-2})\\ & \triangleq(\delta_i)^{T-2}(\frac{1}{2}U^{\top}_{T-2} P^{i*}_{T-1} U_{T-2} +  U_{T-2}^{\top}S^{i*}_{T-1} + w^{i*}_{T-1})+V_{T}^{i*}(U^*_{T-1})\\
    &= \frac{1}{2} (\delta_i)^{T-2} (G_{T-1}U_{T-1}^* - l^i_{T-1})^{\top} Q^i (G_{T-1}U_{T-1}^* - l^i_{T-1}) \\
    &+\sum\limits_{j=1}^N \frac{1}{2}  (\delta_i)^{T-2} (K^{j*}_{T-1} U_{T-2} + L^{j*}_{T-1})^{\top} R^{ij} (K^{j*}_{T-1} U_{T-2} + L^{j*}_{T-1}) \\
    &+ \frac{1}{2} (\delta_i)^{T-1}(U_{T-1}^*)^{\top} P^{i*}_{T} U_{T-1}^* + (\delta_i)^{T-1}  (U_{T-1}^*)^{\top}S^{i*}_{T} + (\delta_i)^{T-1} w^{i*}_{T},
    \end{aligned}
    \end{equation}
for any $U_{T-2}\in \mathcal{B}_{T_{\text{ini}}} \times R^{m(T-2)}$, where $U_{T-1}^* = \mathbf{K}_{T-1}^{*} U_{T-2} + \mathbf{L}_{T-1}^*$, $\mathbf{K}^*_{T-1} = \begin{bmatrix}
I \\
K^*_{T-1}
\end{bmatrix}, \mathbf{L}_{T-1}^* = \begin{bmatrix}
0 \\
L^*_{T-1}
\end{bmatrix},$ $K_{T-1}^*= \text{col}(K_{T-1}^{1*},...,K_{T-1}^{N*}),$ $L_{T-1}^*= \text{col}(L_{T-1}^{1*},...,L_{T-1}^{N*})$ and $I$ is the  identity matrix. 
\begin{equation}
    \begin{aligned}
    &V^{i*}_{T-1}(U_{T-2})\\&=\min\limits_{u^i_{T-1}} \frac{(\delta_i)^{T-2}}{2}(y_{T-1}-l^i_{T-1})^{\top} Q^i (y_{T-1}-l^i_{T-1})+V_{T}^{i*}(U_{T-1}^*) \\&+ \frac{(\delta_i)^{T-2}}{2}[(u^{i}_{T-1})^{\top} R^{ii}u^{i}_{T-1} + \sum\limits_{j \neq i}(u^{j*}_{T-1})^{\top} R^{ij}u^{j*}_{T-1}]  \\&\triangleq \min\limits_{u^i_{T-1}} H^{i*}_{T-1}(u^i_{T-1}), i\in \mathbb N. 
    \end{aligned}
\label{proof eq2-8}
\end{equation}
\vspace{-0.5 cm}
\begin{equation}
    \begin{aligned}
    &\frac{\partial H^{i*}_{T-1}}{\partial u^i_{T-1}} (u^{i*}_{T-1}) \\&= (\delta_i)^{T-2} R^{ii} (K_{T-1}^{i*} U_{T-2} + L_{T-1}^{i*})\\
    &+(\delta_i)^{T-2}[(G_{T-1})_{:,u^i_{T-1}}]^{\top} Q^i (G_{T-1} (\mathbf{K}_{T-1}^{*} U_{T-2} + \mathbf{L}_{T-1}^*) - l^i_{T-1})  \\
    &+ (\delta_i)^{T-1}(P^{i*}_{T})_{u^i_{T-1},:} (\mathbf{K}_{T-1}^{*} U_{T-2} + \mathbf{L}_{T-1}^*) + (\delta_i)^{T-1}(S^{i*}_{T})_{u^i_{T-1},:}  = 0
    \end{aligned}
\label{proof eq2-9}
\end{equation}

Similarly, using $U_{T-1} = \mathbf{K}_{T-1} U_{T-2} + \mathbf{L}_{T-1}^*$ and $y_{T-1} = G_{T-1} U_{T-1}$, and noting that $u^i_{T-1}$ is the best response to $u^{j*}_{T-1}$ for all $j \neq i$, it follows that equation (\ref{proof eq2-7}) equals (\ref{proof eq2-8}), which also leads to equation (\ref{proof eq2-9}), holding for any $U_{T-2} \in \mathcal{B}_{T_{\text{ini}}} \times \mathbb{R}^{m(T-2)}$.
By comparing the coefficients of $\mathrm{col}(u_{\text{ini}}, y_{\text{ini}})$, $u_j$ for $j \leq T-2$, and the constant terms on both sides of the equation, we obtain equations (\ref{eq1-1a}), (\ref{eq1-1b}), and (\ref{eq1-2}) at stage $T-1$. 
It then follows from dynamic programming that, by repeating the process above from stage $T-2$ down to stage $1$, we obtain equations (\ref{eq1-1a}), (\ref{eq1-1b}), and (\ref{eq1-2}) for all $t \in \mathbb{T}$ and $i \in \mathbb{N}$. This completes the proof of part (a).

{ \bf (b) Part (I): A Solution $\Rightarrow$ An FNE.} Second, we prove that for any solution $\{K^{i*}_t, L^{i*}_t,$  $P^{i*}_t,S^{i*}_t, w^{i*}_t| t\in \mathbb T, i\in \mathbb N\}$ to equations (\ref{eq1-1a}), (\ref{eq1-1b}), and (\ref{eq1-2} - \ref{eq1-5}), the strategy profile $\{u^{i*}_t=K^{i*}_t  U_{t-1} + L^{i*}_t\mid t\in \mathbb T, i\in \mathbb N\}$ constitutes an FNE of the unknown-dynamics game (\ref{isosystem}, \ref{costfun}). 
Since equations (\ref{eq1-3} - \ref{eq1-5}) hold at stage $T$, let (\ref{eq1-3}), (\ref{eq1-4}), and (\ref{eq1-5}) represent the quadratic, linear, and constant terms with respect to $U_{T-1}$, respectively, and sum them.
We then obtain equation (\ref{proof eq2-1}), which equals the first two rows of (\ref{proof eq2-2}).
Furthermore, since equations (\ref{eq1-1a}), (\ref{eq1-1b}), and (\ref{eq1-2}) hold at stage $T$, summing them similarly yields (\ref{proof eq2-3}). This implies that $K^{i*}_T$ and $L^{i*}_T$ are the zeros of the first derivative of $H_T^{i*}(u^i_T)$ in (\ref{proof eq2-2}) with respect to $u^i_T$. 
Since $Q^i$ is positive semidefinite and $R^{ii}$ is positive definite, the function $H_T^{i*}(u^i_T)$, representing the objective function of player $i$ in (\ref{proof eq2-2}), is strictly convex in $u^i_T$. Thus, player $i$'s best response exists, is unique, and coincides with the unique root of the derivative of $H_T^{i*}(u^i_T)$. 
This shows that $\{u^{i*}_T=K^{i*}_T  U_{T-1} + L^{i*}_T\mid i \in \mathbb N\}$ is an FNE of the game at stage $T$. 

Since (\ref{eq1-3} – \ref{eq1-5}) hold at stage $T-1$, we similarly obtain (\ref{proof eq2-7}). Since (\ref{eq1-1a}), (\ref{eq1-1b}), and (\ref{eq1-2}) are satisfied at stage $T-1$, summing them yields (\ref{proof eq2-9}), which implies that $K^{i*}_{T-1}$ and $L^{i*}_{T-1}$ are the zeros of the first derivative of $H_{T-1}^{i*}(u^i_{T-1})$ with respect to $u^i_{T-1}$. The function $H_{T-1}^{i*}(u^i_{T-1})$ in (\ref{proof eq2-8}) is strictly convex in $u^i_{T-1}$. Thus, player $i$'s best response is unique and equals the unique root of the derivative of $H_{T-1}^{i*}(u^i_{T-1})$. This shows that the strategies $u^{i*}_{T-1} = K^{i*}_{T-1} U_{T-2} + L^{i*}_{T-1}, u^{i*}_T = K^{i*}_T U_{T-1} + L^{i*}_T $ for $ i \in \mathbb{N}$ form an FNE of the game with unknown dynamics at stages $T-1$ and $T$. It then follows from dynamic programming that, by repeating this process from stage $T-2$ down to $1$, we complete the proof of (b). By combining the proofs of parts (a) and (b), we complete the proof of (I).

{\bf (c) Part of (II).} \textbf{Sufficiency.} 
 Assume that there exists at least one solution of (\ref{eq1-1a}), (\ref{eq1-1b}), and (\ref{eq1-2} – \ref{eq1-5}), and that all solutions satisfy the three conditions in part (II) of the theorem. Due to (b), we have at least one FNE. The three conditions imply that $(K^{i*}_t)_{:,u_j}$ and $L^{i*}_t$ are unique for $j=1,\ldots,t-1$, $i \in \mathbb{N}$, $t \in \mathbb{T}$, and for any two different solutions $\{P^{i*}_t, S^{i*}_t, w^{i*}_t, K^{i*}_t, L^{i*}_t \mid t \in \mathbb{T}, i \in \mathbb{N}\}$ and $\{\hat{P}^{i*}_t, \hat{S}^{i*}_t, \hat{w}^{i*}_t, \hat{K}^{i*}_t, \hat{L}^{i*}_t \mid t \in \mathbb{T}, i \in \mathbb{N}\}$, we have $(K^{i*}_t)_{:,(u_{\text{ini}}, y_{\text{ini}})} \mathbf{B}(\mathcal{B}_{T_{\text{ini}}}) = (\hat{K}^{i*}t)_{:,(u_{\text{ini}}, y_{\text{ini}})} \mathbf{B}(\mathcal{B}_{T_{\text{ini}}})$, which leads to  
\begin{equation}
K^{i*}_t U_{t-1} + L^{i*}_t = \hat K^{i*}_t U_{t-1} + \hat L^{i*}_t, \forall U_{t-1} \in \mathcal{B}_{ini} \times R^{m(t-1)}, t\in \mathbb T, 
\label{eq.equiv}
\end{equation}
 the equation indicates that all FNEs are the same in terms of the inputs $u^i_t$. 
\textbf{Necessity.} 
Assume that there is a unique FNE for the game. Due to part (a), there exists at least one solution of (\ref{eq1-1a}), (\ref{eq1-1b}), and (\ref{eq1-2} – \ref{eq1-5}). Since the FNE is unique, if there are multiple solutions, they must generate the same FNE by the method in part (b). Thus, for any two different solutions of (\ref{eq1-1a}), (\ref{eq1-1b}), and (\ref{eq1-2} – \ref{eq1-5}), namely $\{P_t^{i*}, S_t^{i*}, w_t^{i*}, K_t^{i*}, L_t^{i*} \mid t \in \mathbb{T}, i \in \mathbb{N}\}$ and $\{\hat{P}_t^{i*}, \hat{S}_t^{i*}, \hat{w}_t^{i*}, \hat{K}_t^{i*}, \hat{L}_t^{i*} \mid t \in \mathbb{T}, i \in \mathbb{N}\}$, we have (\ref{eq.equiv}). This directly shows that the three conditions must be satisfied.
\end{proof}

\section{Proof of Theorem \ref{theorem 3}}
\label{Proof of Theorem 3}

\begin{proof}
We first prove (\ref{eq1-1} - \ref{eq1-5}) has a unique solution. At stage $T$, we begin with $P^{i*}_{T+1}, S^{i*}_{T+1}, w^{i*}_{T+1}=0$.  (\ref{eq1-1}) is equivalent to 
$\tilde{H}_T(P_{T+1}^{*}) K_{T}^{*} = \tilde{g}_T(P_{T+1}^{*})$. Since  $|\tilde{H}_T(P_{T+1}^{*})| \neq 0$, we have 
$K_{T}^{*} = \tilde{H}_T(P_{T+1}^{*})^{-1} \tilde{g}_T(P_{T+1}^{*}).$
(\ref{eq1-2}) is equivalent to 
$\tilde{H}_T(P_{T+1}^{*}) L_{T}^{*} = \tilde{g}'_T(S_{T+1}^{*})$, we similarly obtain  $L_{T}^{*} = \tilde{H}_T(P_{T+1}^{*})^{-1} \tilde{g}'_T(S_{T+1}^{*}).$
Then we can directly obtain $P^{i*}_T, S^{i*}_T, w^{i*}_T$ from $K^*_{T}, L_{T}^{*}, P^{i*}_{T+1}, S^{i*}_{T+1}, w^{i*}_{T+1}$ by (\ref{eq1-3} - \ref{eq1-5}). 
At stage $T-1$, for given $P^{i*}_T, S^{i*}_T, w^{i*}_T$, and $|\tilde{H}_{T-1}(P_{T}^{*})| \neq 0$, we similarly obtain $K_{T-1}^{*} = \tilde{H}_{T-1}(P_{T}^{*})^{-1} \tilde{g}_{T-1}(P_{T}^{*}),$ 
$L_{T-1}^{*} = \tilde{H}_{T-1}(P_{T}^{*})^{-1} \tilde{g}'_{T-1}(S_{T}^{*}).$
By repeating the process at stages $T-2,\ldots,2,1$, we obtain the unique solution of (\ref{eq1-1} - \ref{eq1-5}). 
Due to (\ref{eq1-1}) and $\mathbf{K}_t=[ \mathbf{K}_{t_{:,(u_{\text{ini}},y_{\text{ini}})}},\mathbf{K}_{t_{:,u_{1}}},...,\mathbf{K}_{t_{:,u_{t-1}}}]$, we have
\begin{align*}
&[(G_t)_{:,u^i_t}]^{\top} Q^i G_t  \mathbf{K}_{t_{:,(u_{\text{ini}},y_{\text{ini}})}} + R^{ii}  K_{t_{:,(u_{\text{ini}},y_{\text{ini}})}} 
+ \delta_i(P^i_{t+1})_{u^i_t,:}\mathbf{K}_{t_{:,(u_{\text{ini}},y_{\text{ini}})}} = 0,\\
&[(G_t)_{:,u^i_t}]^{\top} Q^i G_t  \mathbf{K}_{t_{:,u_{j}}} + R^{ii} K_{t_{:,u_{j}}} 
+\delta_i(P^i_{t+1})_{u^i_t,:}\mathbf{K}_{t_{:,u_{j}}} = 0, ~j=1,...,t-1, 
\end{align*}
and the first equation multiply by $\textbf B (\mathcal{B}_{T_{\text{ini}}})$ is (\ref{eq1-1a}), the second equation is (\ref{eq1-1b}). Thus, (\ref{eq1-1a}, \ref{eq1-1b}, \ref{eq1-2} - \ref{eq1-5}) has at least one solution. 
By part (I)(b) in Theorem \ref{theorem 2}, the game has at least one FNE.
\end{proof}

\section{Proof of Theorem \ref{theorem 4}}

\label{Proof of Theorem 4}

\begin{proof}
Since Assumption 1 in \cite{2025finitestrategy} holds, it follows from Proposition 1 in \cite{2025finitestrategy} that the $T$-stage known-dynamics game has a unique FNE, for $T=1,2,...$ Here we denote the $T$-stage known-dynamics game's unique FNE by $u^{i*}_t(T)(x_t) = \overline K^{i*}_t(T)x_t, i\in \mathbb N, t=1,..., T$, which is different with the notation in \cite{2025finitestrategy}. 
It follows from Theorem \ref{theorem 1} and Remark \ref{remark} that the FNE $u^{i*}_t(T)(x_t) = K^{i*}_t(T)U_{t-1}$ of the $T$-stage unknown-dynamics game satisfies
$$\overline K^{i*}_t(T) A^{t-k-1}B = (K^{i*}_t(T))_{:,u_{k}}, k=1,...,t-1, i \in \mathbb N, t \in \mathbb T.$$ 
Since Lemma 2 in \cite{2025finitestrategy} holds, we have $\lim\limits_{T\rightarrow +\infty} \overline K^{i*}_t(T)$ $= \overline K^{i*}, i\in \mathbb N$, for any given $t\in N_+$, and $u^{i*}_t(x_t) = \overline K^{i*} x_t$ forms the FNE associated with the limiting matrices in the infinite-horizon known-dynamics game. This leads to $\lim\limits_{T\rightarrow +\infty} (K^{i*}_t(T))_{:,u_{k}} = \lim\limits_{T\rightarrow +\infty} \overline K^{i*}_t(T) A^{t-k-1}B = \overline K^{i*} A^{t-k-1}B,k=1,...,t-1.$

 Since Theorem 1 in \cite{2025finitestrategy} holds, we have $\lim\limits_{T_m \to +\infty} \tilde{J}^i(x_1) = J^i$ for all $i \in \mathbb{N}$, where $T_m = \min(T_1, \ldots, T_N)$. Here, $\tilde{J}^i$ denotes the total cost of player $i$ in the infinite-horizon known-dynamics game when using the finite-horizon strategy of ``watching $T_i$ steps into the future and moving one step now,'' i.e., $u^i_t(x_t) = \overline{K}^{i*}_1(T_i) x_t$ for $t=1,2,\ldots$ and $i \in \mathbb{N}$. Meanwhile, $J^i(x_1)$ denotes the total cost of player $i$ under the FNE of the infinite-horizon known-dynamics game.
By Theorem \ref{theorem 1} and and Remark \ref{remark}, we have $\overline{K}^{i*}_1(T) x_1 = K^{i*}_1(T) \mathrm{col}(u_{\mathrm{ini}}, y_{\mathrm{ini}})$ for all $T=1,2,\ldots$ when $x_1$ is matched with $(u_{\mathrm{ini}}, y_{\mathrm{ini}})$. In the infinite-horizon unknown-dynamics game, if each player $i$ adopts the same finite-horizon strategy $u^i_t = K^{i*}_1(T_i) U_0(t)$ for $i \in \mathbb{N}$ and $t=1,2,\ldots$, then player $i$’s total cost is $\tilde{J}^i(u_{\mathrm{ini}}, y_{\mathrm{ini}})$. Since $x_1$ corresponds to $(u_{\mathrm{ini}}, y_{\mathrm{ini}})$, we have $u^{i*}_t = \overline{K}^{i*}_1(T_i) x_1 = K^{i*}_1(T_i) U_0(1)$ for all $i \in \mathbb{N}$, and thus $x_2$ is still matched with  $U_0(2)$ due to identical inputs at stage 1. Similarly, for all $t \ge 1$, $\overline{K}^{i*}_1(T_i) x_t = K^{i*}_1(T_i) U_0(t)$ holds and $x_t$ is matched with  $U_0(t)$. Therefore, $\tilde{J}^i(u_{\mathrm{ini}}, y_{\mathrm{ini}}) = \tilde{J}^i(x_1)$ for all $i \in \mathbb{N}$, and we have $\lim\limits_{T_m \to +\infty} \tilde{J}^i(u_{\mathrm{ini}}, y_{\mathrm{ini}}) = \lim\limits_{T_m \to +\infty} \tilde{J}^i(x_1) = J^i$.

As shown in the proof above, for $i \in \mathbb N$, we have $\lim\limits_{T\rightarrow +\infty} \overline K^{i*}_1(T)$ $= \overline K^{i*}, i\in \mathbb N$. Thus, for any matched pair $(U_0(t), x_t)$, we have
$$ u^i_t=\overline K^{i*} x_t=\lim\limits_{T\rightarrow +\infty} \overline K^{i*}_1(T)x_t=\lim\limits_{T\rightarrow +\infty} K^{i*}_1(T)U_0(t)=K^{i*}U_0(t).$$ 
Therefore, the strategy profile $\{ u^{i*}_t=K^{i*}U_0(t)\mid i\in \mathbb N, ~t=1,2,...\} $
    is an FNE of the infinite-horizon unknown-dynamics game since $u^{i*}_t(x_t) = \overline K^{i*} x_t$ forms an FNE of the infinite-horizon known-dynamics game.  
Due to Remark \ref{remark} and (\ref{strategy matrices equation2}), we have 
$\overline K^{i*}_1(T) x_1 =  K^{i*}_1(T)_{:,u_{\text{ini}}}u_{\text{ini}}+K^{i*}_1(T)_{:,y_{\text{ini}}}y_{\text{ini}},$
 where $x_1$ is matched with $u_{\text{ini}},y_{\text{ini}}$. Since system (\ref{isosystem}) is controllable, for $e_1 = \text{col}(1,0,\dots,0)\in \mathbb R^n$, ..., $e_n = \text{col}(0,\dots,0,1)\in \mathbb R^n$,  there exist $u_{\text{ini}}(k),y_{\text{ini}}(k) \in \mathcal{B}_{T_{\text{ini}}}$ (with a permutation of the element in $\mathcal{B}_{T_{\text{ini}}}$) is matched with $x_1=e_k, k=1,2,...,n$. We have 
$$\overline K^{i*}_1(T) = \overline K^{i*}_1(T) [e_1,...,e_n]= K^{i*}_1(T) \begin{bmatrix}u_{\text{ini}}(1) & \dots & u_{\text{ini}}(n)\\
y_{\text{ini}}(1) & \dots & y_{\text{ini}}(n)\end{bmatrix}.$$ 

Let $\begin{bmatrix}u_{\text{ini}}(1) & \dots & u_{\text{ini}}(n)\\
y_{\text{ini}}(1) & \dots & y_{\text{ini}}(n)\end{bmatrix} = Z$, since $\lim\limits_{T \rightarrow +\infty} K^{i*}_1(T) = K^{i*}$, we have $ \overline K^{i*}=\lim\limits_{T \rightarrow +\infty} \overline K^{i*}_1(T)=K^{i*}Z , i \in \mathbb N.$ 
Let $\epsilon=\max\limits_{i\in \mathbb N} \| \overline K^{i*}_1(T_i) - \overline K^{i*} \|_2 $, then we have $\lim\limits_{T \rightarrow +\infty} \epsilon = 0$ by Lemma 2 in \cite{2025finitestrategy}. Since  $\|A+\sum\limits_{j=1}^N B^j \overline K^{j*}\|_2  < 1$ by Assumption 1.(iii) in \cite{2025finitestrategy}, there exists an integer $\tilde T \in N_+$ such that when $T_i > \tilde T$ for all $i \in \mathbb{N}$, $\|A+\sum\limits_{j=1}^N B^j \overline K^{j*}\|_2 + (\sum\limits_{j=1}^N \|B^j\|_2) \epsilon < 1$ and $\theta_{i1}  + \theta_{i2} \epsilon + \theta_{i3} \epsilon^2 < 2\theta_{i1} $. It follows from Theorem 1 in \cite{2025finitestrategy} and $\tilde{J}^i(x_1)=\tilde{J}^i(u_{\text{ini}},y_{\text{ini}})$ 
that 
$$
|\tilde{J}^i(u_{\text{ini}},y_{\text{ini}})- J^i| \leq \frac{1}{2} \|x_1\|_2^2  (\theta_{i1}  + \theta_{i2} \epsilon + \theta_{i3} \epsilon^2) \frac{\epsilon}{1-\delta_i}\leq  \|x_1\|_2^2  \theta_{i1}   \frac{\epsilon}{1-\delta_i}.$$ 
Since $
\epsilon =\max\limits_{i\in \mathbb N} \| \overline K^{i*}_1(T_i) - \overline K^{i*} \|_2 =\max\limits_{i\in \mathbb N} \|K^{i*}_1(T_i)Z-K^{i*}Z\|_2 \leq \max\limits_{i\in \mathbb N} \|K^{i*}_1(T_i)-K^{i*}\|_2 \|Z\|_2$, 
we have 
\begin{equation*}
\begin{aligned}
|\tilde{J}^i(u_{\text{ini}},y_{\text{ini}})- J^i| \leq  \|x_1\|_2^2  \theta_{i1}  \frac{\epsilon}{1-\delta_i} \leq  \|x_1\|_2^2  \theta_{i1}  \frac{\|Z\|_2}{1-\delta_i} \max\limits_{i\in \mathbb N} \|K^{i*}_1(T_i)-K^{i*}\|_2.
\end{aligned}
\end{equation*}
Thus, when $T_m > \tilde T$, there exists a constant $\mathcal{C}=\|x_1\|_2^2  \theta_{i1}  \frac{\|Z\|_2}{1-\delta_i}$ such that the following inequality holds.  This completes the proof. 
\begin{equation*}
\begin{aligned}
|\tilde{J}^i(u_{\text{ini}},y_{\text{ini}})- J^i| \leq  \mathcal{C} \max\limits_{i\in \mathbb N} \|K^{i*}_1(T_i)-K^{i*}\|_2 .
\end{aligned}
\end{equation*} 
\end{proof}

\bibliographystyle{siamplain}
\bibliography{main}

\end{document}